\documentclass[a4paper,10pt]{article}
\usepackage{graphicx}
\usepackage{color}
\usepackage{amsmath, amsthm, slashed}
\usepackage{times}
\usepackage{cancel}
\usepackage{enumerate}
\usepackage{amssymb}
\usepackage{rotating}
\usepackage{epsfig}
\usepackage{verbatim}
\usepackage{rotating}
\usepackage{graphicx}

\newtheorem{theorem}{Theorem}[section]
\newtheorem{conjecture}{Conjecture}[section]
\newtheorem{corollary}{Corollary}[section]
\newtheorem{proposition}{Proposition}[section]
\newtheorem{lemma}{Lemma}[section]
\newtheorem{remark}{Remark}[section]

\newcommand{\nabb}{\mbox{$\nabla \mkern-13mu /$\,}}

\title{Decay properties of Klein-Gordon fields on 
Kerr-AdS spacetimes}

\author{Gustav Holzegel\footnote{Princeton University, Department of Mathematics, Fine Hall, Washington Road, Princeton, NJ 08544, United States.}
\, and Jacques Smulevici\footnote{Laboratoire de Math\'ematiques, Universit\'e Paris-Sud 11, b\^at. 425, 91405 Orsay, France.}
}

\begin{document}
\maketitle
\begin{abstract}
This paper investigates the decay properties of solutions to the massive linear wave equation $\Box_g \psi + \frac{\alpha}{l^2} \psi =0$ for $g$ the metric of a Kerr-AdS spacetime satisfying $|a|l<r_+^2$ and $\alpha<\frac{9}{4}$ satisfying the Breitenlohner Freedman bound. We prove that the non-degenerate energy of $\psi$ with respect to an appropriate foliation of spacelike slices decays like $\left(\log t^\star\right)^{-2}$. 
%The proof relies on a frequency decomposition of the solution and interpolation between multiplier estimates proven for each frequency range. 
Our estimates are expected to be sharp from heuristic and numerical arguments in the physics literature suggesting that general solutions will only decay logarithmically. The underlying reason for the slow decay rate can be traced back to a stable trapping phenomenon for asymptotically Anti de Sitter black holes which is in turn a consequence of the reflecting boundary conditions %for $\psi$ 
at null-infinity. 
\end{abstract}

% With AMS-LaTeX, \maketitle follows the abstract

%%      ---------------------------------------------------------------------
%%      ------------------- TABLE OF CONTENTS (OPTIONAL) --------------------
%%      ---------------------------------------------------------------------

%% ***** IF YOUR PAPER IS OVER 40 PAGES AND YOU WISH TO HAVE A TABLE
%% ***** OF CONTENTS, PLEASE UNCOMMENT THE FOLLOWING LINE

 \tableofcontents

%%      ---------------------------------------------------------------------
%%      ---------------------------- BODY OF PAPER --------------------------
%%      ---------------------------------------------------------------------

%%      Please input or insert the body of your paper here.

\section{Introduction}
In a sequence of recent papers \cite{Holzegelwp, HolzegelAdS,gs:lwp,gs:stab}, the study of boundedness and decay properties of 
waves on asymptotically Anti-de Sitter (AdS) black hole spacetimes was initiated. The analysis focused on the dynamics of the massive wave equation
\begin{align} \label{mwe}
\Box_g \psi + \frac{\alpha}{l^2} \psi = 0 
\end{align}
for $g$ an asymptotically AdS spacetime with cosmological constant $\Lambda= -3l^{-2}$ and $\alpha$ satisfying $\alpha<\frac{9}{4}$, the Breitenlohner-Freedman bound. The well-posedness of this equation on such backgrounds was established in \cite{Vasy2, Holzegelwp,claude:wp} under suitable boundary conditions at null-infinity. (See also \cite{Bachelot, Ishibashi} for earlier results on exactly AdS.) The uniform boundedness of solutions to (\ref{mwe}) was proven in \cite{HolzegelAdS} for Kerr-AdS spacetimes, provided the angular momentum per unit mass of the black hole is not too large.\footnote{In fact, the result was shown to hold for a class of $C^1$ perturbations of these spacetimes.} The proof relied on three crucial ingredients:
\begin{enumerate}[(a)]
\item the existence of a globally causal Killing vectorfield on the black hole exterior (the Hawking-Reall vectorfield \cite{HawkingReall}) for the range $|a|l < r_+^2$ where $r_+=r_+\left(M,a,l\right)$ is the $r$-value at the horizon in Boyer-Lindquist coordinates, \label{it1}
\item the exploitation of weighted Hardy inequalities allowing one to absorb the (wrong-signed) mass-term in the range $0<\alpha< \frac{9}{4}$, \label{it2}
\item the redshift effect near the event horizon coupled with the Hawking-Reall vectorfield allowed to prove boundedness \emph{without} the construction of an integrated decay (Morawetz) estimate. This goes back to an idea of Dafermos and Rodnianski in the asymptotically-flat, Schwarzschild-case, cf.~\cite{Mihalisnotes}. \label{it3}
\end{enumerate}
The ingredients (\ref{it1}) and (\ref{it3}) are particularly remarkable, when compared with the situation in the asymptotically flat case, where the boundedness statement for solutions of the massless wave equation on Kerr with $a \ll M$ is already intimately connected with the decay properties of the solution  \cite{DafRodKerr}: The absence of a globally causal Killing field implies that techniques based purely on spacetime vectorfields are insufficient to control the solution and instead microlocal analysis is required (or higher order differential operators, see \cite{AndBlue}). Interestingly, integrated \emph{decay} estimates had to be constructed for superradiant frequency ranges to prove \emph{boundedness} of the solution. The Kerr-AdS-case, on the other hand, is (precisely because of (\ref{it1}) and (\ref{it3})) more similar to the asymptotically-flat Schwarzschild case, where the boundedness and the decay statements can also be entirely decoupled. 

Regarding (\ref{it1}) above, we also remark that in the range $|a| l > r_+^2$ solutions are expected to grow because superradiant effects are continuously amplified by the reflecting AdS boundary \cite{Cardosoinstability1, Cardosoinstability2, blackholebomb}.

While the decay of solutions on asymptotically-flat Kerr black holes has recently been understood for the subextremal range $|a|<M$ \cite{DafRodlargea} after a period of intense research over the last ten years \cite{DafRod, DafRod2, Mihalisnotes, DafRodlargea, DafRodsmalla, DafRodnew, Toha1, Toha2, AndBlue}, the question of decay of solutions to (\ref{mwe}) remained unanswered in \cite{HolzegelAdS}. Very recently the authors, in the context of a spherically symmetric non-linear problem \cite{gs:lwp, gs:stab}, proved in particular that spherically symmetric solutions of (\ref{mwe}) decay exponentially on the Schwarzschild-AdS black hole exterior provided $\alpha\leq2$ holds. The latter result exploited a version of the redshift near null-infinity, as well as the absence of any trapping phenomena in spherical symmetry.

Beyond spherically symmetric solutions of the wave equation, however, the picture changes radically. This may be anticipated at the heuristic level from an analysis of geodesic flow (which is the high frequency or ``geometric optics" approximation for the propagation of waves) on the Schwarzschild-AdS manifold (i.e.~$a=0$). As in the asymptotically flat case, the radial equation for the null-geodesics sees a potential $V\left(r\right)$ which exhibits a local maximum at $r=3M$ corresponding to the familiar photon sphere. Consequently, ingoing geodesics which do not have sufficient energy to surmount the potential will be reflected off towards null-infinity. While in the asymptotically flat case these (now outgoing) geodesics eventually leave the spacetime through null-infinity, in the asymptotically AdS case, they are reflected back into the spacetime. This is a stable trapping phenomenon, which constitutes a serious obstacle to prove decay. The above reasoning does not take ``tunneling-effects" (i.e.~the uncertainty principle for the wave equation) into account. Indeed, an exponentially small part of the energy will tunnel through the potential and leave the black hole exterior through the event horizon. From these heuristics one may expect a logarithmic decay rate and this is indeed what we shall prove below.  

\subsection{Logarithmic Decay}
We now state our main results directing the reader to Section \ref{sec:norms} for a precise definition of the energy norms and functional spaces involved in Theorem \ref{pendia}. The statement below also makes reference to a foliation of the Kerr-AdS black hole exterior by spacelike slices, $\Sigma_{t^\star}$, which  will be defined in Section \ref{sec:KerrAdS}. For the reader familiar with the formalism of Penrose diagrams, the foliation may be read off from Figure 1.

\begin{theorem} \label{theorem1}
Let $\left(\mathcal{M},g \right)$ be a Kerr-AdS spacetime with parameters $\left(M,a,l\right)$ for which
both the Hawking-Reall condition $r_+^2 > |a| l$ and the regularity condition $|a|<l$ hold.
Let $\psi$ be a $CH_{AdS}^2$ solution of (\ref{mwe}) with $-\infty < \alpha< \frac{9}{4}$ on this background.  Assume moreover that one of the following three conditions on the parameters hold:
\begin{enumerate}
\item $\alpha \le 1$  ,
\item $1 < \alpha \le 2$ and $|a| < \frac{l}{2}$  ,
\item $a$ is sufficiently small depending only on $M$, $l$ and $\alpha$. % $\left(\frac{r_+}{l}\right)^2 < 8$, where $r_+$ denote the value of $r$ on the event horizon, and
\end{enumerate}
Then, for the foliation of spacelike slices $\Sigma_{t^\star}$ (intersecting the event horizon) defined in Section \ref{Foliations}, we have the decay estimate  
\begin{align} \label{es:maines}
\| \psi \|_{H^1_{AdS}\left(\Sigma_{t^\star}\right)} \lesssim \left(\log \ t^\star \right)^{-1} \left(E_2 \left[\psi \right]\right)^\frac{1}{2}
\end{align}
for any $t^\star \geq 2$, where $\lesssim$ allows a constant depending only on the parameters $M$, $a$, $l$ of the background spacetime and on $\alpha$.
\end{theorem}

By what have become standard techniques (commuting with $T$, the redshift and angular momentum vectorfields, cf.~\cite{Mihalisnotes}, \cite{HolzegelAdS}), one obtains this $\left(\log \ t^\star\right)^{-1}$-decay for all higher order norms and also pointwise estimates via Sobolev embedding. In particular, one has

\begin{corollary}
Under the assumptions of Theorem 1 and with $\psi$ in addition in $CH^3_{AdS}$ one has the pointwise bound
\begin{align}
|\psi| \lesssim \left(\log \ t^\star \right)^{-1} \left( \left(E_2\left[\psi\right]\right)^\frac{1}{2} + \left(E_2\left[\partial_{t^\star} \psi\right]\right)^\frac{1}{2} \right) \, .
\end{align}
\end{corollary}

We remark that we have made no serious attempt to optimize the conditions on the parameter $a$ in items 2 and 3 of Theorem \ref{theorem1}. They can most likely be improved with further work.

\begin{figure}[h] \label{pendia}
\[
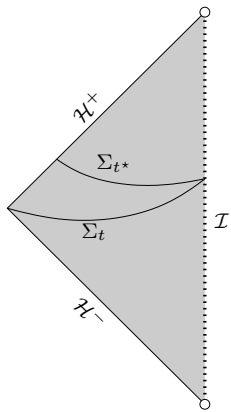
\]
\caption{The spacelike hypersurfaces $\Sigma_t$ and $\Sigma_{t^\star}$.}
\end{figure}

\subsection{Absence of time-periodic solutions}
Obviously, Theorem \ref{theorem1} excludes time-periodic solutions. This situation is in sharp contrast with that of pure $AdS$, where periodic solutions may be easily constructed, cf.~\cite{Breitenlohner}. It turns out that the absence of periodic solutions can be proven in a more elementary fashion which is why we include and prove this as a separate theorem.
\begin{theorem} \label{theorem2}
Let $\left(\mathcal{M},g \right)$ be a Kerr-AdS spacetime with parameters $\left(M,a,l\right)$ such that $|a| l < r_+^2$ and $|a|<l$. Let $\psi$ be a solution of (\ref{mwe}) in the class $CH^1_{AdS}$, with $\alpha < 9/4$, with $\psi$ of the form $\psi=e^{i\omega t}\hat{\psi}$, where $\omega \in \mathbb{R}$ and $\hat{\psi}$ is constant in $t$.
Assume moreover that one of the three conditions on the parameters given in Theorem \ref{theorem1} hold. Then $\psi=0$.
\end{theorem}
The above theorem implies the absence of time-periodic solutions in $CH^1_{AdS}$: Indeed, given a solution $\psi$ in $CH^1_{AdS}$ which is periodic in $t$ with period $\frac{1}{\omega}$ for some non-zero $\omega \in \mathbb{R}$, we have the expansion $\psi = \sum_{n=-\infty}^\infty e^{i \omega_n t} \hat{\psi}_n$, with $\omega_n = 2\pi \cdot n \cdot \omega$. By Theorem \ref{theorem2} all the $\hat{\psi}_n$ are zero.
\begin{remark}
In view of the recent results of  \cite{GHHCMW}, Theorem \ref{theorem2} actually holds without any restrictions on the parameters other than $|a| l < r_+^2$, $|a|<l$ and $\alpha<9/4$. See Section \ref{noperiodic}.
\end{remark}
\subsection{Quasi-normal modes}
The question of decay for solutions of wave equations in the presence of trapping has a long tradition in the context of the so-called \emph{obstacle problem} in Minkowski space. Here, one studies the propagation of waves outside of a compact object. The analogue of trapped null-geodesics are piecewise straight lines reflected off the boundary of the obstacle but (due to the geometry of the obstacle) remaining in a compact region of space. For this problem, one can prove that, as in the setting studied here, a logarithmic local energy decay estimate always holds, even in the presence of stable trapping. In \cite{nb:del}, this decay result was obtained via an analysis of the so-called \emph{resonances}, also known as quasi-normal modes in the physics literature. Resonances are the poles of the meromorphic continuation of the (truncated) resolvent of the associated elliptic operator. In the black hole setting, such techniques have been very successfully applied on asymptotically
de Sitter manifolds \cite{Melrose, Haefner, Vasy, VasyDyatlov, Dyatlov1, Dyatlov2}.
 
Finally, the issue of locating the quasi-normal modes is also a popular theme in the physics literature \cite{Hubeny, Festuccia}. For the case of the wave equation on Schwarzschild-AdS, it was found numerically in \cite{Festuccia} that the imaginary part of the quasi-normal modes approach the axis exponentially fast as the angular momentum mode $\ell$ goes to infinity. This is suggestive of general solutions decaying at best logarithmically.

\subsection{The non-linear stability problem}
Let us relate our results to the problem of non-linear stability for asymptotically AdS spacetimes. As mentioned above, there is no decay for the wave equation on pure AdS, in fact, one can construct an infinity of periodic solutions on this background. For the associated non-linear problem, that is the evolution under the Einstein equations of perturbations of Cauchy data for AdS (supplemented with appropriate boundary conditions), this is highly suggestive of instability. This fact was first pointed out in a talk of M.~Dafermos \cite{Newtontalk}, and suggested independently (using different arguments) by the results of \cite{AndAdS}. Recently, there has been some interesting numerical work \cite{BizonAdS} in support of AdS being unstable.

Turning from pure AdS to black hole spacetimes, one has the result of \cite{gs:stab} proving asymptotic stability of Schwarzschild-AdS with respect to \emph{spherically symmetric} perturbations in the context of the Einstein-Klein-Gordon system. With this in mind, one may expect a similar result to hold for general perturbations of Schwarzschild-AdS, or more generally, Kerr-AdS spacetimes. On the contrary, the analysis of this paper suggests that this belief is likely to be misconceived: While the trapping mechanism is absent in the spherically symmetric case (allowing one to prove exponential decay in this symmetry class), it will, in general, prevent the linear solution from decaying faster than logarithmically. For the non-linear problem, the logarithmic decay rate of the linear solution will be too weak to close a standard bootstrap argument. Thus one may conjecture:
\begin{conjecture}
Kerr-AdS spacetimes are dynamically unstable for generic perturbations.
\end{conjecture}

\subsection{Outline}
In Section \ref{sec:KerrAdS}, we present the family of Kerr-AdS Lorentzian manifolds. Sections \ref{sec:norms}  and \ref{se:wpbs} contain the functional framework for the Klein-Gordon field as well as a review of previous material (well-posedness and boundedness results) required later in the paper. In Section \ref{se:pre}, we exploit the separation of variables for the wave equation (\ref{mwe}) to localize the solutions in various frequency ranges associated to both the time and the angular variables. This localization allows to replace (\ref{mwe}) by an ordinary differential equation with a (frequency dependent) potential. The relevant properties of this potential are also proven in this section. 

Our decay result is based on an integrated decay estimate for the low frequency part, which is derived via the construction of microlocalized energy currents which generate positive spacetime terms (pioneered in the black hole context in \cite{Mihalisnotes, DafRodsmalla, DafRodlargea}). These currents are introduced in Section \ref{se:mct} and applied in Sections \ref{phiflatt2} and \ref{phiflatt}. A key ingredient is the use of \emph{exponential} weights in these currents, with the constant in the exponential depending on the frequency cut-off. In Section \ref{se:sephif}, we collect our frequency space estimates to obtain an integrated decay estimate on the low frequency part of the solution. The high frequency part of the solution is estimated in Section \ref{se:esph} using simply the boundedness statement: This estimate gains a smallness factor (after commutation with angular derivatives) from the fact that the solution is supported on high frequencies. Finally, we conclude the proof of the main theorem in Section \ref{se:pth1} by interpolating between the high and low frequency estimates. 

A special case of Kerr-AdS spacetimes are the Schwarzschild-AdS spacetimes. In this case, we can improve the estimates for the low frequency part of the solution to establish exponential decay of the energy (with the constant depending on the frequency cut-off). This is done in Section \ref{se:srsc}. Section \ref{noperiodic} contains the proof of Theorem \ref{theorem2}. This section can in fact be read independently of Sections \ref{phiflatt2}--\ref{se:srsc}. Finally, Appendix \ref{ap:suf} contains some standard estimates associated to the frequency decomposition.

%\subsection{Final Remark}
%Several estimates below depend on the sign of $\alpha$. In order not to burden the reading of the paper, we shall treat in detail the case $\alpha \ge 0$, which is the difficult case. The simpler case $\alpha < 0$ may be treated via minor modifications of the arguments given below. Thus, from now on, we assume that $\alpha \ge 0$. 
%
%
%
%
%
%
%
%
\section{The Kerr-AdS family of spacetimes} \label{sec:KerrAdS}
The following presentation of Kerr-AdS spacetimes is strongly inspired by the presentation of the Kerr spacetimes found in \cite{DafRodsmalla}.
\subsection{The fixed manifold with boundary $\mathcal{R}$}
Let $\mathcal{R}$ denote the manifold with boundary
$$
\mathcal{R}=[0,\infty) \times \mathbb{R} \times \mathbb{S}^2.
$$
We define standard coordinates $y^*$ for $\mathbb{R}^+$, $t^*$ for $\mathbb{R}$ and $(\theta,\phi)$ for $\mathbb{S}^2$. This defines a coordinate system on $\mathcal{R}$, which is global modulo the well-known degeneration of the spherical coordinates. 

We define the event horizon $\mathcal{H}^+$ to be the boundary of $\mathcal{R}$
$$
\mathcal{H}^+=\partial \mathcal{R}=\{ y^*=0 \}.
$$ 
We denote by $T$ and $\Phi$ the vector fields $\partial_{t^*}$ and $\partial_{\phi}$, where the latter is to be understood as the extension of the coordinate vector field to all of $\mathbb{S}^2$ in the usual way.
Let $\phi_{\tau}$ denote the one-parameter family of diffeomorphisms generated by $T$.
The manifold $\mathcal{R}$ will coincide with the domain of outer communication of the black hole spacetimes including the future event horizon $\mathcal{H}^+$. 
\subsection{The parameter space}
The Kerr-AdS family of spacetimes depends on three parameters $(a,M,l)$, which are respectively the angular momentum per unit mass, the mass and a quantity related to the cosmological constant $\Lambda$ via $\Lambda= -\frac{3}{l^2}$. The parameters are assumed to satisfy $M>0$ as well as $|a| < l$
(the latter is needed to exclude naked singularities). Specifically, we define the parameter space $\mathcal{P}$ to be
$$
\mathcal{P}:=\{ (a,M,l) \in \mathbb{R} \times (0,\infty)^2: |a| < l \}.
$$

For any $(a,M,l) \in \mathcal{P}$, denote by $r_+(a,M,l)$ the largest real root of the polynomial 
\begin{equation}\label{def:deltam}
\Delta_-(r)= \left(r^2+ a^2\right) \left(1+ \frac{r^2}{l^2}\right) - 2Mr.
\end{equation}

\subsection{Kerr-AdS-star coordinates}
Let $r(a,M,l,y^*)$ be a smooth real valued function defined on $\mathcal{P} \times [0,\infty)$ such that $r_{|(a,M,l) \times [0,\infty)}$ is a diffeomorphism onto $[r_+(a,M,l),\infty)$ which coincides with $y^*$ for $y^* > 3M$. For each parameter tuple $(a,M,l) \in \mathcal{P}$, the collection $(t^*,r(a,M,l,y^*),\phi,\theta)$ determines a coordinate system on $\mathcal{R}$, global up to the degeneration of the spherical coordinates, which we call \emph{Kerr-AdS-star} coordinates.

\subsection{The tortoise coordinate $r^*$}
Given parameters $(a,M,l) \in \mathcal{P}$, we define the coordinate $r^*$ on $\{ r > r_+ \}$ as
%\subsection{Tortoise coordinates}
%Finally, it will sometimes be convenient to work with a coordinate system $\left(t, r^\star, \theta, \tilde{\phi}\right)$, which is obtained from Boyer-Lindquist coordinates as follows. If $r_+$ denotes the largest root of $\Delta_-=0$ then $r^*$ is defined by
\begin{eqnarray}
\frac{d r^*}{dr}=\frac{r^2+a^2}{\Delta_-(r)} \textrm{ \ \ \ \ , \ \ \ }
r^*(r=+\infty)=\pi/2,
\end{eqnarray}
with $\Delta_-(r)$ as in \eqref{def:deltam}.
Note that $r^*=-\infty$ at the horizon. We will use the notation $R^\star_\infty=\frac{\pi}{2}$. 

We also define $R^\star_{-\infty}$ to be some value of $r^\star$ close to the horizon, i.e. $R^\star_{-\infty}$ is a large negative number, and we will eventually consider the limit $R^\star_{-\infty} \rightarrow -\infty$ in our estimates.

%The latter coordinate system only cover the subset of $\mathbb{S}^2$ corresponding to $(0,\pi) \times (0,2\pi}$, but we extend these functions and the associated vector fields in the usual manner. 

\subsection{Boyer-Lindquist coordinates}
Finally, given parameters $(a,M,l) \in \mathcal{P}$, we define Boyer-Lindquist coordinates $(t,r,\theta, \tilde{\phi})$ from the Kerr-AdS-star coordinates as follows:

%We will also consider a coordinate system $\left(t^\star,r,\theta,\phi\right)$, which is regular on the event horizon. It is related to the Boyer-Lindquist set by the transformations
\begin{equation}
t = t^* - A\left(r\right) \textrm { \ \ \ and \ \ \ } \tilde{\phi} = \phi - B\left(r\right)
\end{equation}
where 
\begin{equation}
\frac{dA}{dr} = \frac{2Mr}{\Delta_- \left(1+\frac{r^2}{l^2}\right)} \textrm{ \ \ \ and \ \ \ } \frac{dB}{dr} = \frac{a \left(1-\frac{a^2}{l^2}\right)}{\Delta_-}
\end{equation}
and $A$ and $B$ vanishing at infinity.

The Boyer-Lindquist coordinate system only covers $\mathrm{int}\left(\mathcal{R}\right)$, i.e. the set $\{ r > r_+ \}$, which is the exterior of the black hole excluding the event horizon. Boyer-Lindquist coordinates are of particular relevance because it is with respect to these coordinates that the wave equation separates.

\subsection{The Kerr-AdS metric for fixed $(a,M,l)$}
We may now introduce the Kerr-AdS metric as the unique smooth extension to $\mathcal{R}$ of the tensor given on the Boyer-Lindquist chart by:
\begin{align}\label{eq:metricbl}
g_{KAdS} = \frac{\Sigma}{\Delta_-} dr^2 + \frac{\Sigma}{\Delta_\theta} d\theta^2 + \frac{\Delta_\theta \left(r^2+a^2\right)^2 - \Delta_- a^2 \sin^2 \theta}{\Xi^2 \Sigma} \sin^2 \theta d\tilde{\phi}^2 \nonumber \\
- 2 \frac{\Delta_\theta \left(r^2+a^2\right)-\Delta_-}{\Xi \Sigma}a \sin^2 \theta \ d\tilde{\phi} dt - \frac{\Delta_- - \Delta_\theta a^2 \sin^2 \theta}{\Sigma} dt^2
\end{align}
where $\Delta_-$ is defined by \eqref{def:deltam} and 
\begin{align}
\Sigma &= r^2 + a^2 \cos^2 \theta, 
\ \ \ \ \ 
\Delta_\pm = \left(r^2+ a^2\right) \left(1+ \frac{r^2}{l^2}\right) \pm 2Mr
 \\
\Delta_\theta &= 1 - \frac{a^2}{l^2} \cos^2 \theta,
 \ \ \ \ \ 
\Xi = 1 - \frac{a^2}{l^2} \, .
\end{align}
Note the components of the inverse in these coordinates:
\begin{align}
g^{tt} = - \frac{\left(r^2+a^2\right)^2}{\Sigma \Delta_-} + \frac{a^2 \sin^2 \theta}{\Sigma \Delta_\theta} \, ,
\end{align}
\begin{align}
g^{t\tilde{\phi}} = \frac{\Xi \left(r^2+a^2\right)}{\Delta_- \Sigma} a - \frac{\Xi}{\Delta_\theta \Sigma} a 
\textrm{\ , \ \ \ \ } g^{\tilde{\phi} \tilde{\phi}} = \frac{\Xi^2}{\Sigma \Delta_\theta} \frac{1}{\sin^2 \theta} - \frac{\Xi^2 a^2}{\Sigma  \Delta_-} \, ,
\end{align}
\begin{align}
g^{rr} = \frac{\Delta_-}{\Sigma}, \ \ \ \ \ g^{\theta \theta} = \frac{\Delta_\theta}{\Sigma}, \ \ \  \ \ \sqrt{g} = \frac{\Sigma}{\Xi}\sin \theta \, .
\end{align}

%\subsection{Regular coordinates}
That the tensor \eqref{eq:metricbl} indeed extends to a smooth metric on $\mathcal{R}$ is clear from examining its form in Kerr-AdS-star coordinates for which we have:
%The new metric coefficients become
\begin{eqnarray}
g_{\theta \theta} &=& \frac{\Sigma}{\Delta_\theta} \textrm{ \ \ , \ \ }
g_{t^\star t^\star} = g_{tt}  \textrm{ \ \ , \ \ } \\
g_{\phi \phi} &=& g_{\tilde{\phi} \tilde{\phi}} \textrm{ \ \ , \ \ }
g_{t^\star \phi} = g_{t \tilde{\phi}}  \textrm{ \  , \ \ } \\
g_{rr} &=& \frac{1}{\Sigma \left(1+\frac{r^2}{l^2}\right)^2}\left(\Delta_+ - a^2 \sin^2 \theta \left(k_0 + \frac{\Sigma}{l^2}\right) \right)  \textrm{ \  , \ \ } \\
g_{\phi r} &=& -\frac{a \sin^2 \theta}{\Xi \Sigma k_0} \left[\Xi \Sigma + 2Mr \right]  \textrm{ \  , \ \ } \\
g_{t^\star r} &=& \frac{1}{\Sigma k_0} \left(2Mr - \frac{a^2}{l^2} \sin^2\theta \Sigma \right) \, .
\end{eqnarray}
with $k_0 = 1 + \frac{r^2}{l^2}$. The inverse components are
\begin{eqnarray}
g^{t^\star t^\star} &=& -\frac{\Xi \Delta_+ + a^2 \sin^2 \theta \left(-k_0 \Xi + \frac{2Mr}{l^2}\right)}{k_0^2 \Delta_\theta \Sigma} \, , \\
g^{\theta \theta} &=& \frac{\Delta_\theta}{\Sigma} \textrm{ \ \ , \ \ }
g^{\phi \phi} = \frac{\Xi^2}{\Delta_\theta \Sigma \sin^2 \theta} \textrm{ \ \ , \ \ }
g^{t^\star \phi} = \frac{a \Xi}{k_0 l^2 \Delta_\theta} \, , \\
g^{rr} &=& \frac{\Delta_-}{\Sigma}  \textrm{ \ \ , \ \ \ }
g^{\phi r} = \frac{a \Xi}{\Sigma} \textrm{ \ \ , \ \ \ \ \ \ \ \ \ \ \ \ \ \ }
g^{t^\star r} = \frac{2Mr}{k_0 \Sigma} \, .
\end{eqnarray}
The unit normal of a constant $t^\star$ slice is
\begin{equation}
n_\Sigma = \sqrt{-g^{t^\star t^\star}} \partial_{t^\star} - \frac{g^{t^\star r}}{\sqrt{-g^{t^\star t^\star}}}\partial_r - \frac{g^{t^\star \phi}}{\sqrt{-g^{t^\star t^\star}}} \partial_\phi \textrm{ \ \ \ , \  $g(n_\Sigma,n_\Sigma)=-1$} \, ,
\end{equation}
with the determinant of the metric induced on constant $t^\star$ slices being
\begin{equation} \label{indmeas}
\sqrt{\det h} = \sqrt{\det g_{t^\star=const}} = \frac{\Sigma}{\Xi} \sin \theta \left(\sqrt{-g^{t^\star t^\star}}\right) \, .
\end{equation}
Finally, note that one (formally) recovers the familiar asymptotically flat Kerr metric from (\ref{eq:metricbl}) in the limit as $l \rightarrow \infty$, while for $a \rightarrow 0$ one recovers Schwarzschild-AdS. The metric $g_{KAdS}$ is seen to be asymptotically AdS after a change of coordinates as shown explicitly in appendix B of \cite{Henneaux}.
\subsection{Coordinate Vectorfields and Isometries} \label{covecf}
Note the following relation between the coordinate vectorfields in regular coordinates $\left(t^\star,r,\theta,\phi\right)$ and in Boyer-Lindquist tortoise coordinates $\left(t,r^\star,\theta,\tilde{\phi}\right)$:
\begin{align} %\label{derrel}
& \ \  \ \ \  \ \ \  \ \ \ \  \ \ \  \ \ \  \ \  \  \ \ \  \ \ \ \  \ \ \  \ \ \  \ \  \partial_{t^\star} = \partial_t \ \ \ , \ \ \ \partial_{\phi} = \partial_{\tilde{\phi}} \, , \ \  \nonumber \\
& \partial_r =  \frac{1}{1+\frac{r^2}{l^2}} \partial_t - \frac{r-r_+}{\Delta_-} \cdot \frac{a \Xi \left(r+r_+\right)}{r_+^2 + a^2} \partial_{\tilde{\phi}} + \frac{r^2 + a^2}{\Delta_-} \left(-\partial_t + \partial_{r^\star} - \frac{a\Xi}{r_+^2 + a^2} \partial_{\tilde{\phi}}\right) \nonumber \, .
\end{align}
This implies in particular that the vectorfield
\begin{align} \label{Zdef}
Z = \frac{1}{\Delta_-} \left(-\partial_t + \partial_{r^\star} - \frac{a\Xi}{r_+^2 + a^2} \partial_{\tilde{\phi}}\right)
\end{align}
extends continuously to the event horizon and that $\Delta_- Z \psi$ vanishes linearly in $r-r_+$ on the event horizon.

The vectorfields $T=\partial_{t}$ and $\Phi=\partial_{\tilde{\phi}}$ are both Killing fields. An important linear combination is the  so-called Hawking-Reall vectorfield
\begin{align}
K = T + \frac{a \Xi}{r_+^2+a^2} \Phi \, ,
\end{align}
which in addition to being Killing is globally causal on the black hole exterior, provided that $|a|l<r_+^2$ holds. Note that $K$ becomes null on the event horizon, while it is strictly timelike on the exterior.
\subsection{Foliations} \label{Foliations}
We denote by $\Sigma_{t^\star}$ (respectively $\Sigma_t$) the hypersurfaces of constant $t^\star$ (respectively constant $t$), cf.~Figure \ref{pendia} in the introduction. We write $\Sigma_0$ for the slice $\Sigma_{t^\star=0}$, which is where (for convenience) we will prescribe data for (\ref{mwe}). We also make the convention that the Greek letter $\tau$ is always associated with the $t^\star$ foliation.
\section{The norms} \label{sec:norms}
Let $\slashed{g}$ and $\slashed{\nabla}$ denote the induced metric and the covariant derivative on the spheres $\mathbb{S}^2_{t^\star, r}$ of constant $t^\star$ and $r$ in $\mathcal{R}$. 

 We write $|\slashed{\nabla} ... \slashed{\nabla} \psi|^2=\slashed{g}^{AA^\prime} ... \slashed{g}^{BB^\prime} \slashed{\nabla}_A ... \slashed{\nabla}_{B} \psi  \slashed{\nabla}_{A^\prime} ... \slashed{\nabla}_{B^\prime} \psi$ to denote the induced norm on the spheres.
We define the following non-degenerate energy norms for the scalar field $\psi$, cf.~\cite{Holzegelwp}:
\begin{eqnarray}
\| \psi \|^2_{H^{0,s}_{AdS}\left(\Sigma_{t^\star}\right)} &=&  \int_{\Sigma_{t^\star}} r^s  \psi^2 r^2 dr \sin \theta d\theta d{\phi}\nonumber \\
\| \psi \|^2_{H^{1,s}_{AdS}\left(\Sigma_{t^\star}\right)} &=& 
\int_{\Sigma_{t^\star}} r^s \left[ r^2 \left(\partial_r \psi \right)^2 + | \nabb \psi |^2 + \psi^2 \right] r^2 dr \sin \theta d\theta d{\phi}\nonumber \\
\| \psi \|^2_{H^{2,s}_{AdS}\left(\Sigma_{t^\star}\right)} &=& \| \psi \|^2_{H^{1,s}_{AdS}\left(\Sigma_{t^\star}\right)} \nonumber \\
&&\hbox{} + \int_{\Sigma_{t^\star}} r^s \bigg[  r^4 \left(\partial_r \partial_r \psi \right)^2 
+r^2 | \nabb \partial_r \psi|^2 +| \nabb \nabb \psi |^2 \bigg] r^2 dr \sin \theta d\theta d{\phi}\nonumber
\end{eqnarray}
Higher order norms may be defined similarly. We denote by $H^{k,s}_{AdS}(\Sigma_{t^*})$, the space of functions $\psi$ such that $ \nabla^i \psi \in L^2_{loc}(\Sigma_{t^*})$ for $i=0,..., k$ and $\| \psi \|^2_{H^{k,s}_{AdS}\left(\Sigma_{t^\star}\right)} < \infty$.
We denote by $CH_{AdS}^{k,s}$ the set of functions $\psi$ defined on $\mathcal{R}$, such that 
\begin{align}
\psi \in &\bigcap_{q=0,..,k} C^q\left(\mathbb{R}_{t^\star};H^{k-q,s_q}_{AdS}(\Sigma_{t^\star})\right) \nonumber \\
&\textrm{where $s_k = -2$, $s_{k-1} = 0$ and $s_j = s$ for $j=0, ... , k-2$. }\nonumber
\end{align}
When $s=0$, we will feel free to drop the $s$ in the notation, i.e. $H^{k,0}_{AdS}:=H^{k}_{AdS}$ and $CH_{AdS}^{k}:=CH_{AdS}^{k,0}$. 

We also define the densities
\begin{align} \label{endensdef}
  e_1 \left[\psi\right] &=  \frac{1}{r^2} \left(\partial_{t^\star} \psi\right)^2  + r^2  \left(\partial_{r} \psi\right)^2 + | \slashed{\nabla} \psi|^2 + \psi^2 \, ,
\nonumber \\
e_2 \left[\psi\right] &=   e_1 \left[\psi\right] + e_1 \left[\partial_{t^\star} \psi \right] + \sum_{i=1}^3 e_1 \left[\Omega_i \psi\right]  + r^4 \left(\partial_r \partial_r \psi \right)^2   \, ,
\end{align}
which the reader can convert to Boyer-Lindquist coordinates using the formulae of Section \ref{covecf}. Here the $\Omega_i$ denote the standard basis of angular momentum operators on the unit sphere in the coordinates $(\theta,\phi)$. Finally, let us define (and impose finiteness of) the initial energies
\begin{align} \label{endef}
\begin{split}
  E_1 \left[\psi\right] &=  \|\psi\|^2_{H^{1,0}_{AdS}\left(\Sigma_0\right)} + \|\partial_{t^\star} \psi\|^2_{H^{0,-2}_{AdS}\left(\Sigma_0\right)} < \infty  \, ,  \\
E_2 \left[\psi\right] &=  \|\psi\|^2_{H^{2,0}_{AdS}\left(\Sigma_0\right)} + \|\partial_{t^\star} \psi\|^2_{H^{1,0}_{AdS}\left(\Sigma_0\right)}+ \sum_{i} \| \Omega_i \psi \|^2_{H^1_{AdS\left(\Sigma_0\right)}} \\ & \ \ \ \ \ + \|\partial_{t^\star} \partial_{t^\star} \psi\|^2_{H^{0,-2}_{AdS}\left(\Sigma_0\right)} < \infty \,.
\end{split}
\end{align}

\section{Well-posedness and boundedness statement} \label{se:wpbs}
We review in this section previous results concerning Klein-Gordon fields in Kerr-AdS spacetimes which will be used later in this paper. 

From Definition 5.1 of \cite{Holzegelwp} we recall that an \emph{$H^2$-initial data set} $\left(u,v,w\right)$ for \eqref{mwe} on a slice $\Sigma_{t^\star}$ of constant $t^\star$, say $t^\star=0$ for convenience, is determined by specifying freely a pair  $v \in H^{1,0}_{AdS}\left(\Sigma_0\right)$, $w \in H^{0,-2}_{AdS}\left(\Sigma_0\right)$, which then uniquely determines a function $u \in H^{2,s}_{AdS} \left(\Sigma_0\right)$ (for any $s < \sqrt{9-4\alpha}$) by solving an elliptic equation. The following statement is proven in \cite{Holzegelwp}:

\begin{theorem}[Well-posedness in the energy class, \cite{Holzegelwp}] \label{theo:lwp}
Fix a Kerr-AdS spacetime with parameters $\left(M,l,a\right)$ and let $\alpha < 9/4$. Let $(u,v,w)$ be an $H^2_{AdS}$-initial data set on the spacelike slice $\Sigma_{0}$.
Then, there exists a unique solution $\psi$ of \eqref{mwe} such that $\psi|_{t^\star=0}=u$, $\partial_t \psi |_{t^\star=0}=v$ and such that $\psi$ belongs to the class of $CH^2_{AdS}$ functions (unique within $CH^1_{AdS}$). 
\end{theorem}

As usual, higher regularity is propagated and one can obtain the analogous statement for $H_{AdS}^k$-norms with arbitrary $k$. 

%Since we are not concerned with regularity (but rather with obtaining global uniform estimates on $\psi$) here, there is no loss in considering \emph{smooth} solutions.
 
We collect below also the main vector-field estimates leading to the uniform boundedness statement for solutions to (\ref{mwe}). To state them in a form useful for the remainder of the paper, we recall the notion of the energy momentum tensor associated with $\psi$ satisfying (\ref{mwe}),
\begin{align}
\mathbb{T}_{\mu \nu}\left[\psi\right]  = \partial_\mu \psi \partial_\nu \psi - \frac{1}{2} g_{\mu \nu} \left(g^{\gamma \delta} \partial_\gamma \psi \partial_\delta \psi - \frac{\alpha}{l^2} \phi^2 \right) \, .
\end{align}
For a given spacetime vectorfield $X$, we define the energy current
\begin{align}
J^X_\mu \left[\psi\right] = \mathbb{T}_{\mu \nu}\left[\psi\right] X^\nu \, ,
\end{align}
which satisfies the identity
\begin{align}
\nabla^\mu \left(J^X_\mu \left[\psi\right] \right)=  \nabla^\mu \left(\mathbb{T}_{\mu \nu}\left[\psi\right] X^\nu\right) = \mathbb{K}_{\mu \nu}\left[\psi\right] {}^{(X)} \pi^{\mu \nu} \, ,
\end{align}
with ${}^{(X)}\pi^{\mu \nu} = \frac{1}{2} \left( \nabla^{\mu} X^\nu + \nabla^{\nu} X^{\mu} \right)$ being the deformation tensor of $X$. The latter vanishes for $X$ a Killing field. 

For the following theorem and for the remainder of the paper, we denote by $C=C_{M,l,a,\alpha}$ a constant depending only on the fixed parameters in the problem. The symbols $\lesssim$ and $\gtrsim$ are used to state inequalities which hold modulo such constants.

\begin{theorem} \label{theo:bndness}
Given parameters $M>0$ and $l>0$ there exists an $a_{max}>0$ such that for all Kerr-AdS spacetimes with parameters $M$, $l$, $a$ with $|a|<a_{max}$, the following is true: Any solution $\psi$ of (\ref{mwe}) with $\alpha<\frac{9}{4}$ arising via Theorem \ref{theo:lwp} from initial data defined on $\Sigma_0$ satisfies
\begin{itemize}
\item Conservation law: 
\begin{align} \label{conslaw}
 \int_{\Sigma_{t^\star}} J^K_{\mu}\left[\psi\right] n^\mu + \int_{\mathcal{H}\left(t^\star,0\right)} \left[K\left(\psi\right)\right]^2 =  \int_{\Sigma_{0}} J^K_{\mu}\left[\psi\right] n^\mu \, .
\end{align}
\item Control of non-degenerate norms: There exists a globally timelike vectorfield $N$ (with $-g\left(N,N\right)>c>0$) on the black hole exterior with\footnote{Recall that if $0<\alpha<\frac{9}{4}$, then $J^N_\mu \left[\psi\right] n^\mu$ is actually not positive definite. However, the integral over $t^\star$ slices of that expression is, cf.~\cite{HolzegelAdS}.}
\begin{align}
\|\psi\|^2_{H^{1,0}_{AdS}\left(\Sigma_{t^\star}\right)} + \|\partial_{t^\star} \psi\|^2_{H^{0,-2}_{AdS}\left(\Sigma_{t^\star}\right)} \lesssim \int_{\Sigma_{t^\star}} J^N_{\mu}\left[\psi\right] n^\mu \, .
\end{align}
The analogous statement holds with $N$ replaced by $K$ on the right- and norms degenerating near the horizon on the left-hand side implying that all of the terms appearing in (\ref{conslaw}) are non-negative.

\item Uniform boundedness statement:
\begin{align} \label{unibs}
\int_{\Sigma_{t^\star}} J^N_{\mu}\left[\psi\right] n^\mu \lesssim  \int_{\Sigma_{0}} J^N_{\mu}\left[\psi\right] n^\mu
\end{align}
for all $t^\star$ slices to the future of $\Sigma_0$. Similar estimates hold for higher order energies.
\item
Redshift estimate: There exists $3M>r_1>r_0>r_+$ such that the estimate
\begin{align} \label{reold}
\int_{\mathcal{R}\left(\tau_1,\tau_2\right) \cap \{r\leq r_0\}} J^N_{\mu}\left[\psi\right] n^\mu \lesssim \int_{\mathcal{R}\left(\tau_1,\tau_2\right) \cap \{r_0 \leq r\leq r_1\}} J^N_{\mu}\left[\psi\right] n^\mu \nonumber \\ + C_{M,l,a, \alpha} \int_{\mathcal{R}\left(\tau_1,\tau_2\right) \cap \{r\leq r_1\}} \phi^2 +  \int_{\Sigma_{\tau_1}} J^N_{\mu}\left[\psi\right] n^\mu 
\end{align}
holds for any $\tau_2 \geq \tau_1 \geq 0$.
\end{itemize}
\end{theorem}
In \cite{HolzegelAdS}, the estimate (\ref{reold}) was crucial to obtain the boundedness statement, as the spacetime terms on the right hand side were shown to be controlled by the spacetime integral of $J^K_{\mu}\left[\psi\right] n^\mu$. Here, we will use (\ref{reold}) in conjunction with an integrated decay estimate in the interior for a certain range of frequencies. See Proposition \ref{localN} in Section \ref{sec:hozimprove}.

Analogous boundedness statements can be obtained for higher norms. For instance, recalling (\ref{endensdef}) and (\ref{endef}) it is not hard to show the uniform estimate
\begin{align} \label{bndadded}
\int_{\Sigma_{t^\star}} e_2 \left[\psi\right] r^2 \sin \theta dr d\theta d\phi \lesssim E_2\left[\psi\right] \, .
\end{align}

\section{Preliminaries} \label{se:pre}
In the remainder of the paper, we study the solution $\psi$ arising from initial data prescribed on $\Sigma_{0}$ as described in Theorem \ref{theo:lwp}. Without loss of generality we will assume $\psi$ to be smooth and obtain statements in $H^k_{AdS}$ by density. We also recall the notation $\lesssim$ and $\gtrsim$ from Section \ref{se:wpbs}.
\subsection{Time cut-off} \label{timecutoff}
Let $\chi\left(z\right)$ be a smooth real cut-off function which is equal to zero for $z \leq 0$ and equal to $1$ for $z \geq 1$. For any $\tau \in \mathbb{R}$, we define $\chi_{\tau} \left(t^\star\right) = \chi \left(t^\star\right) \chi \left(\tau-t^\star\right)$ and finally
\begin{align}
\psi^{\tau}\left(t^\star, \cdot \right) = \chi_{\tau} \left(t^\star\right) \ \psi \left(t^\star, \cdot \right)
\end{align}
which is supported in $0 \leq t^\star \leq \tau$ only. 
%To simplify the notation, we denote $\psi_{\scis}$ by $\Psi$.  
If $\psi$ is a solution of the wave equation \eqref{mwe}, then $\psi^{\tau}$ satisfies the inhomogeneous equation
\begin{align} \label{bpf}
\Box_{g_{KAdS}} \psi^{\tau} + \frac{\alpha}{l^2} \psi^\tau= F := 2 \nabla^\alpha \chi^\star \nabla_{\alpha} \psi + \left(\Box_{g_{KAdS}}  \chi^{\tau}\right) \psi \, ,
\end{align}
%\begin{align}
%\Box_{g_{KAdS}} \Psi = F := 2 \nabla^\alpha \xi_{\tau^\prime} \nabla_{\alpha} \psi + \left(\Box_{g_{KAdS}}  \xi_{\tau^\prime}\right) \psi
%\end{align}
with $F$ supported in the two strips
\begin{align}
\{ 0 \leq t^\star \leq 1 \} \cup \{\tau - 1 \leq t^\star \leq \tau \} \, .
\end{align}
Note that 
\begin{align}
r^2 |F|^2 \lesssim \frac{1}{r^2} \left(\partial_{t^\star} \psi\right)^2 + r^2 \left(\partial_r \psi\right)^2 + | \slashed{\nabla} \psi |^2 + \psi^2  \, ,
\end{align}
and hence in view of the support of $F$ and Theorem \ref{theo:bndness}, for any $\tau$,
\begin{align} \label{errorloc}
\int_{\mathcal{R}\left(0,\tau\right)} \left(r^2 |F|^2 \right) r^2  \sin \theta dt^\star dr d\theta d\phi \lesssim \|\psi\|^2_{H^1_{AdS}\left(\Sigma_0\right)} \, .
\end{align}
Note also that
\begin{align}
 \int_{\Sigma_{t^\star}} J^N_{\mu}\left[\psi^\tau \right] n^\mu \lesssim  \int_{\Sigma_{t^\star}} J^N_{\mu}\left[\psi\right] n^\mu \lesssim \|\psi\|^2_{H^1_{AdS}\left(\Sigma_0\right)}
\end{align}
and
\begin{align} \label{hozcon}
\int_{\mathcal{H}\left(t^\star,0\right)} \left[K\left(\psi^\tau\right)\right]^2 \lesssim \|\psi\|^2_{H^1_{AdS}\left(\Sigma_0\right)} \, .
\end{align}
To see (\ref{hozcon}), convert the zeroth order error-term in the strips arising from the cut-off into a spacetime term using the radial derivative of $\psi$. Finally, use the boundedness statement.

In view of the time cut-off, for any $r>r_+$, $\psi^{\tau}$ is now Schwartz in the $t$-variable (with values in $H^{2}_{AdS}$) and we can take the (inverse) Fourier transform,
\begin{align} \label{foutime}
\psi^{\tau} \left(t,r,\theta,\tilde{\phi}\right) = \frac{1}{\sqrt{2\pi}} \int_{-\infty}^\infty e^{-i\omega t} \widehat{\psi^{\tau}} \left(\omega,r,\theta,\tilde{\phi}\right) d\omega \, .
\end{align}
We would like to decompose $\widehat{\psi^{\tau}}$ further using the separability properties of the wave equation. In this process, we will encounter certain self-adjoint operators on $\mathbb{S}^2$, whose eigenfunctions are related to the oblate spheroidal harmonics. We will study the spectral properties of these operators first. 
\subsection{The (modified) oblate spheroidal harmonics} \label{se:sphhar}
Define the $L^2\left(\sin \theta d\theta d\tilde{\phi}\right)$-self-adjoint operator $P$ acting on $H^1\left(\mathbb{S}^2\right)$-complex valued functions (see section 7 of \cite{DafRodsmalla} for a more detailed discussion) as
\begin{eqnarray}\label{def:pop}
-P\left(\xi\right) f &=& \frac{1}{\sin \theta} \partial_\theta \left(\Delta_\theta \sin \theta \partial_\theta f \right) + \frac{\Xi^2}{\Delta_\theta} \frac{1}{\sin^2 \theta} \partial_{\tilde{\phi}}^2 f \nonumber \\ 
&&\hbox{}+ \Xi \frac{\xi^2}{\Delta_\theta} \cos^2 \theta f - 2 i  \xi \frac{\Xi}{\Delta_\theta} \frac{a^2}{l^2} \cos^2 \theta \  \partial_{\tilde{\phi}}f \, \, ,
\end{eqnarray}
with $\xi = a \omega$. We also define the operator $P_{\alpha}$, which is equal to
\begin{equation} \label{def:palpha}
P_{\alpha} \left(\xi\right) = \left\{
\begin{array}{rl} 
P \left(\xi\right) + \frac{\alpha}{l^2}a^2 \sin^2 \theta & \text{if } \alpha > 0 \, , \\
P\left(\xi\right) + \frac{|\alpha|}{l^2}a^2 \cos^2 \theta & \text{if } \alpha \leq 0 \, .
\end{array} \right. 
\end{equation}
For $l \rightarrow \infty$ the operator reduces to the oblate spheroidal operator on $\mathbb{S}^2$ considered in \cite{DafRodsmalla}. If also $a=0$, we retrieve the Laplacian on the round sphere. Standard elliptic theory justifies the definitions (cf.~appendix \ref{ap:suf})
\begin{align} \label{def:spherop}
P\left(\xi\right) \textrm{ \ has eigenvalues } \tilde{\lambda}_{m \ell}\left(\xi\right) \textrm{ with eigenfunctions $\tilde{S}_{m\ell} \left(a\xi, \cos \theta \right) e^{im \tilde{\phi}}$}, \nonumber \\
P_\alpha \left(\xi\right) \textrm{ \ has eigenvalues }\lambda_{m \ell}\left(\xi\right) \textrm{ with eigenfunctions $S_{m\ell} \left(a\xi, \cos \theta \right) e^{im \tilde{\phi}}$}.
\end{align}
To analyse the spectral properties of these operators, we introduce\footnote{Recall that in the asymptotically flat case, the range $0\leq m \omega \leq m \omega_+$ corresponds to the superradiant frequencies. In our case, there is no ``superradiance'' in view of the properties of the Hawking-Real vector field. Nevertheless it will be useful to rewrite the potential in terms of $\omega-\omega_+$, see below.}
\begin{align} \label{def:orop}
\omega_+ = \frac{m a \Xi}{r_+^2 +a^2} \ \ \ \ \ , \ \ \ \ \omega_r\left(r\right) = \omega_r = \frac{m a \Xi}{r^2 +a^2}  \, .
\end{align}
\begin{lemma}[Spectral properties of $P_{(\alpha)}+\xi^2$] \label{spheroidal}
The $\lambda_{m \ell}$ satisfy the estimates
\begin{align}
\lambda_{m\ell} + \xi^2 &\geq \Xi^2 \cdot |m| \left(|m| +1 \right) \, ,  \\
\lambda_{m\ell} + \xi^2 &\geq \Xi^2 \cdot |m| \left(|m| +1 \right) + \frac{m^2 \Xi^2 a^4}{\left(r_+^2 + a^2\right)^2} - C_{a,l} |m| |\omega -\omega_+| \, .
\end{align}
The same estimates hold replacing $\lambda_{m\ell}$ by $\tilde{\lambda}_{m \ell}$.
\end{lemma}
\begin{proof}
We will establish the estimates for $\tilde{\lambda}_{m\ell}$ (i.e.~the operator $P\left(\xi\right)$). The same proof goes through simply neglecting the additional non-negative term present in $P_{\alpha}\left(\xi\right)$.
Recalling (\ref{def:spherop}) and writing shorthand $f= \tilde{S}_{m \ell} \left(\xi, \cos \theta\right)e^{im\tilde{\phi}}$, we find
\begin{align}
P\left(\xi\right) f + \xi^2 f = \nonumber \\
 -\frac{1}{\sin \theta} \partial_\theta \left(\Delta_\theta \sin \theta \partial_\theta f \right) - \frac{\Xi^2}{\Delta_\theta} \frac{1}{\sin^2 \theta} \partial_{\tilde{\phi}}^2 f + \frac{\xi^2}{\Delta_\theta} \sin^2 \theta f - 2 m \xi \frac{\Xi}{\Delta_\theta} \frac{a^2}{l^2} \cos^2 \theta f \nonumber \\
% = -\frac{1}{\sin \theta} \partial_\theta \left(\Delta_\theta \sin \theta \partial_\theta f \right) - \frac{\Xi^2}{\Delta_\theta} \frac{1}{\sin^2 \theta} \left( 1 - \frac{a^4}{l^4} \cos^4 \theta \right) \partial_{\tilde{\phi}} ^2 f  + P_c \left(f\right) \nonumber\\
= -\frac{1}{\sin \theta} \partial_\theta \left(\Delta_\theta \sin \theta \partial_\theta f \right) - \Xi^2\frac{1}{\sin^2 \theta}  \partial_{\tilde{\phi}}^2 f  +P_c \left(f\right)
= \tilde{P}\left(\xi\right) f + P_c \left(f\right) \, , \nonumber
\end{align}
where
\begin{align}
P_c =  \frac{1}{\Delta_\theta} \left( \xi \sin \theta - m \Xi \frac{a^2}{l^2} \frac{\cos^2 \theta}{\sin \theta}\right)^2 +  \frac{\Xi^2 m^2 a^2}{l^2} \frac{\cos^2 \theta}{\sin^2 \theta} \geq 0 \, .
\end{align}
For $P_c$ one can also derive the estimate (use $\sin^{-2} \theta \geq 1 + \cos^2 \theta$, $\Delta_\theta^{-1} \geq 1 + \frac{a^2}{l^2} \cos^2 \theta$ and the $r_+^2 > |a| l$ condition)
\begin{align}
P_c \geq  \frac{m^2 \Xi^2 a^4}{\left(r_+^2 + a^2\right)^2} - C_{a,l} |m| |\omega -\omega_+| \, .
\end{align}
For $\tilde{P}$ we integrate by parts to see that
% \marginpar{We may want to add a little more here based on the min-max principle.}
\begin{align}
\int_{\mathbb{S}^2} \tilde{P} \left(\xi\right) f \cdot \bar{f} = \int_{\mathbb{S}^2} \left[ \Delta_\theta |\partial_\theta f|^2 +  \Xi^2   \frac{1}{\sin \theta} |\partial_{\tilde{\phi}} f|^2 \right] d\theta d\tilde{\phi} \, , \nonumber \\ 
\geq \Xi^2 \int_{\mathbb{S}^2} \left[ |\partial_\theta f|^2 +   \frac{1}{\sin^2 \theta} |\partial_{\tilde{\phi}} f|^2 \right] \sin \theta d\theta d\tilde{\phi} \nonumber
\end{align}
holds for any $f \in H^1\left(\mathbb{S}^2\right)$.
\end{proof}
From Lemma \ref{spheroidal} an application of the Cauchy-Schwarz inequality yields

\begin{corollary} \label{corhelp}
There exists a constant $c_{M,l,a}$ such that
for any $0<\delta<c_{M,l,a}$ the following holds: In the regime $\left(\omega-\omega_+\right)^2 \leq \delta^2 \mathcal{L}:= \delta^2 \left(\lambda_{m\ell} + a^2 \omega^2\right)$ we have
\begin{align}
\lambda_{m\ell} + \xi^2 &\geq\frac{1}{1+ \frac{\delta}{c_{M,l,a}}} \left[ \Xi^2 \cdot |m| \left(|m| +1 \right) + \frac{m^2 \Xi^2 a^4}{\left(r_+^2 + a^2\right)^2} \right] \, .
\end{align}
\end{corollary}
Note that the pre-factor can be chosen as close to $1$ as one desires by choosing $\delta$ appropriately small.

%Denoting the eigenvalues of $P_\alpha$ by $\lambda_{m\ell}$ and the associated eigenfunctions by $S_{m \ell} \left(\xi, \cos \theta\right) e^{im \tilde{\phi}}$, we have from the above:

%\begin{corollary}
%The $\lambda_{m\ell}$ also satisfy the estimates
%\begin{align}
%{\lambda}_{m\ell} + \xi^2 &\geq \Xi^2 \cdot |m| \left(|m| +1 \right) \, , \nonumber \\
%{\lambda}_{m\ell} + \xi^2 &\geq \Xi^2 \cdot |m| \left(|m| +1 \right) + \frac{m^2 \Xi^2 a^4}{\left(r_+^2 + a^2\right)^2} - C_{a,l} |m| |\omega -\omega_+| \, .\nonumber
%\end{align}
%\end{corollary}

\subsection{The separation of variables} \label{se:sepv}
We are ready to obtain the fully separated wave equation. We decompose the $\widehat{\psi^{\tau}}$ of equation (\ref{foutime}) as
\begin{align}
\widehat{\psi^{\tau}}\left(\omega, r , \theta, \tilde{\phi}\right)  &= \sum_{m\ell} \left(\widehat{\psi^{\tau}}\right)^{(a\omega)}_{m\ell} \left(r\right) S_{m\ell} \left(a \omega, \cos \theta\right) e^{im\tilde{\phi}} 
\end{align}
where ($d\tilde{\sigma} = \sin \theta d\theta d\tilde{\phi}$)
\begin{align}
 \left(\widehat{\psi^{\tau}}\right)^{(a\omega)}_{m\ell} \left(r\right) = \frac{1}{\sqrt{2\pi}} \int_{-\infty}^\infty dt \int_{\mathbb{S}^2\left(t,r\right)} d\tilde{\sigma} e^{i\omega t} \ S_{m\ell} \left(a \omega, \cos \theta\right) e^{-im\tilde{\phi}} \  \psi^{\tau}\left(t,r,\theta,\tilde{\phi}\right) \nonumber .
\end{align}
The analogous Fourier-decomposition for the source term $F$ in (\ref{bpf}) yields
\begin{align}
F^{(a\omega)}_{m\ell} \left(r\right)= \frac{1}{\sqrt{2\pi}} \int_{-\infty}^\infty dt \int_{\mathbb{S}^2\left(t,r\right)} d\tilde{\sigma} \ e^{i\omega t} \ S_{m\ell} \left(a \omega, \cos \theta\right) e^{-im\tilde{\phi}} \  F \left(t,r,\theta,\tilde{\phi}\right) \nonumber \, .
\end{align}
After the renormalization
\begin{align}
u^{(a\omega)}_{m\ell} \left(r\right)  = \left(r^2+a^2\right)^\frac{1}{2} \Psi^{(a\omega)}_{m\ell} \left(r\right) \,, 
\end{align}
we finally obtain the equation
\begin{align}
\left[ u^{(a\omega)}_{m\ell} \left(r\right)\right]^{\prime \prime} + \left(\omega^2 - V^{(a\omega)}_{m\ell} \left(r\right) \right) u = H^{(a\omega)}_{m\ell} \left(r\right) \label{eq:uode} \, ,
\end{align}
with the right hand side being equal to
\begin{align}
H^{(a\omega)}_{m\ell} \left(r\right)  = \frac{\Delta_-}{\sqrt{r^2 + a^2}} F^{(a\omega)}_{m\ell} \left(r\right) \,  .
\end{align}
The potential $V^{(a\omega)}_{m\ell} \left(r\right)$ is defined as
\begin{align}
V^{(a\omega)}_{m\ell} \left(r\right) = V^{(a\omega)}_{+, m \ell}\left(r\right) + V^{(a\omega)}_{0, m \ell} \left(r\right) +V_\alpha \left(r\right) \, ,
\end{align}
where
\begin{eqnarray} \label{eq:v+}
V^{(a\omega)}_{+, m \ell}\left(r\right) &=& -\Delta_-^2 \frac{3 r^2 }{\left(r^2 + a^2\right)^4} + \Delta_-\frac{5\frac{r^4}{l^2} + 3r^2 \left(1+\frac{a^2}{l^2}\right) - 4Mr + a^2}{\left(r^2 + a^2\right)^3} \nonumber \\
&=& \left(r^2+a^2\right)^{-\frac{1}{2}} \left(\sqrt{r^2+a^2}\right)^{\prime \prime},\\
V^{(a\omega)}_{0, m \ell} \left(r\right)&=& \frac{\Delta_- \left(\lambda_{m\ell} + \omega^2 a^2\right)-\Xi^2 a^2 m^2 - 2m\omega a \Xi \left(\Delta_- - \left(r^2+a^2\right)\right)}{\left(r^2 + a^2\right)^2}\nonumber \\
V_\alpha \left(r\right) &=&-\frac{\alpha}{l^2} \frac{\Delta_-}{\left(r^2+a^2\right)^2} \left(r^2 + \Theta\left(\alpha\right) a^2 \right) \label{eq:va} \, ,
\end{eqnarray}
with $\Theta \left(\alpha\right)=1$ for $\alpha > 0$ and $\Theta \left(\alpha\right)=0$ if $\alpha\leq 0$ (recall that the $\lambda_{m\ell}$ also depend on $\alpha$ from (\ref{def:palpha})). Note that $V_+$ grows like $\frac{2r^2}{l^4}$ near infinity while the $V_0$-part remains bounded.
\subsection{Properties of the Potential}
We will now study the potential in some detail. We will establish
good positivity properties of the potential near the event horizon for all cases of Theorem \ref{theorem1}. For $\alpha<1$ we in fact obtain \emph{global} (in $r$) positivity estimates. 
For the remaining $\alpha$-range, we will have to allow for a negative contribution in the global estimate.

We will suppress the sub- and superscripts for the potential and its components in this subsection, i.e.~we will write $V$ for $V^{(a\omega)}_{m\ell}$, $V_+$ for $V^{(a\omega)}_{+, m \ell}$, etc.   

\begin{lemma}[Frequency independent properties]  \label{lem:Vplus} 
If any of the three conditions of Theorem \ref{theorem1} holds, there exists $\rho_{+} > r_+$, depending only on $a$, $l$, $M$ and $\alpha$ such that, for $r \in [r_+,\rho_{+}]$,
\begin{align} \label{ineq:Vplusalpha}
V_+ + V_\alpha > C \frac{\Delta_-}{\left(r^2+a^2\right)^2}  \  r^2, 
\end{align}
where $C$ is a constant depending only on $M$, $l$, $a$, $\alpha$. Moreover, again for any of the three conditions of Theorem \ref{theorem1}, we have the global estimate
\begin{align} \label{ida2}
V_+ + V_\alpha + \frac{1}{4l^2} \frac{\Delta_-}{r^2 + a^2} \geq C \frac{\Delta_-}{\left(r^2+a^2\right)^2}  \  r^2 \, .
\end{align}
For $\alpha<1$ the last term on the left hand side of (\ref{ida2}) can be dropped.
\end{lemma} 
\begin{proof}
We have
\begin{align} \label{Vplus}
V_+ = \frac{2\Delta_-}{\left(r^2+a^2\right)^2} \frac{r^2}{l^2}+ \frac{\Delta_-}{\left(r^2+a^2\right)^4} \left(a^2 \Delta_- + \left(r^2-a^2\right) 2Mr\right) \geq 0  \, ,
\end{align}
from which one easily obtains \eqref{ineq:Vplusalpha} globally if $\alpha < 1$. For $\alpha<\frac{9}{4}$ it is straightforward to check that, adding the extra term on the left, (\ref{ida2}) indeed holds under any of the conditions of Theorem \ref{theorem1}.  To obtain the estimate locally near the horizon without the extra term, we collect the leading order terms near the horizon from (\ref{eq:v+}) to find that we need to establish positivity for
\begin{eqnarray}
5\frac{r^4}{l^2} + 3r^2 \left(1+\frac{a^2}{l^2}\right) - 4Mr + a^2-\frac{\alpha}{l^2}(r^2+a^2)^2
\end{eqnarray}
at $r_+$. Using $\Delta_-(r_+)=0$ to evaluate the term $4Mr_+$, the sign of the above expression is the same as that of:
\begin{eqnarray}
3 \frac{r_+^4}{l^2}+r_+^2 \left(1+\frac{a^2}{l^2} \right)-a^2-\alpha \frac{r_+^4}{l^2}-2\alpha \frac{r_+^2}{l^2} a^2-\alpha \frac{a^4}{l^2},
\end{eqnarray}
which is easily seen to be strictly positive under any of the three conditions of Theorem \ref{theorem1}.
\end{proof}
Note that one can obtain the estimate (\ref{ineq:Vplusalpha}) near the horizon also for $\left(V_+ + V_\alpha\right)^\prime$. For $V_0$, we simply remark the following property.
\begin{lemma} \label{lem:V0}
Setting $\mathcal{L}=\lambda_{m \ell} + \omega^2 a^2$, we have
\begin{align} \label{V0-om}
V_0 - \omega^2 = - \left(\omega - \omega_r\right)^2 + \frac{\Delta_-}{\left(r^2+a^2\right)^2} \left(\mathcal{L} - 2m  a \Xi \omega\right) 
\end{align}
with the last term in (\ref{V0-om}) being non-negative.
\end{lemma}
\begin{proof}
The equality in (\ref{V0-om}) is an easy computation while 
the positivity claim is a consequence of the positivity of the operator $P+\xi^2-2m \xi \Xi$, which can be read off from
\begin{eqnarray} %\label{eq:v0pos}
\left(P+\xi^2-2m\xi \Xi\right)f= -\frac{1}{\sin \theta} \partial_\theta \left(\Delta_\theta \sin \theta \partial_\theta f \right)+ \frac{1}{\Delta_{\theta}}\left( \xi \sin \theta-\frac{m \Xi}{\sin \theta} \right)^2f  \, . \nonumber
\end{eqnarray}
\end{proof}
To perform further analysis near the horizon we introduce 
\begin{align} \label{vred}
V_{red} = V + |\omega - \omega_r(r)|^2 - \omega^2 \ , \ \ \ \ \ \ \tilde{V} = V + |\omega - \omega_+|^2 - \omega^2  \, .
\end{align}
From the previous Lemmata, we deduce
\begin{lemma}
If any of the three conditions of Theorem \ref{theorem1} holds, there exists $\rho_+> r_+$, depending only on $a$, $l$, $M$ and $\alpha$, such that for all $r \in [r_+,\rho_+]$, we have:
\begin{align} \label{pol}
V_{red}&\ge C \frac{\Delta_-}{\left(r^2+a^2\right)^2}  \cdot r^2+\frac{\Delta_-}{(r^2+a^2)^2}\left( \mathcal{L}-2ma\Xi \omega \right)  %\quad \mathrm{with\ } f(r_+) > 0,
%V_{red} &\ge 0, \nonumber\\
%V_{red}'&\ge C \frac{\Delta_-}{\left(r^2+a^2\right)^2}  \cdot r^2+\frac{\Delta_-}{(r^2+a^2)^2}\left( \mathcal{L}-2ma\Xi \omega \right),  \nonumber \\%\quad  \mathrm{with\ } g(r_+) > 0,
%\frac{V_{red}'}{V_{red}}& \ge C,
%|V_{red}-\tilde{V}|&=\mathcal{O}(r-r_+) \nonumber
\end{align}
%\marginpar{I think writing that $C$ depends on $r_+$ confuses the reader who could think that the constant becomes worse as the horizon is approached!}
with $C > 0$ a constant depending only on $a$, $l$, $M$, $\alpha$. Moreover, the estimate (\ref{pol}) is global for any of the conditions of Theorem  \ref{theorem1}, provided the term $\frac{1}{4l^2} \frac{\Delta_-}{r^2+a^2}$ is added to the left hand side when $\alpha \geq 1$.
\end{lemma}

In Section \ref{phiflatt2}, we will analyze the behavior of solutions in the regime 
\begin{align} \label{regime1}
|\omega - \omega_+ |^2 \leq \delta^2 \mathcal{L}^\star := \delta^2 \cdot \max \left(\lambda_{ml} + a^2 \omega^2, \Xi^2\right)
\end{align}
for some small $\delta$ to be determined. Note that $\mathcal{L}^\star = \mathcal{L}$ for $m\neq 0$.
Note also that the above range includes the near stationary regime $\omega \sim \omega_+$. For this, we will need the following estimate on $\tilde{V}$:

\begin{lemma} \label{lem:Vtilde} 
If any of the three conditions of Theorem \ref{theorem1} holds, then there exists a $\delta > 0$ 
depending only on $a$, $l$, $M$ and $\alpha$, such that for all frequencies $|\omega-\omega_+|^2 \le \delta^2 \mathcal{L}^\star$  we have the global estimate
\begin{align} \label{eq:global}
\tilde{V} + \frac{1}{4l^2}\frac{\Delta_-}{r^2+a^2} \ge C \frac{\Delta_-}{\left(r^2+a^2\right)^2}\left(m^2 + r^2\right) \, ,
\end{align}
where the second term on the left hand side can be dropped if $\alpha < 1$. 

Moreover, if any of the three conditions of Theorem \ref{theorem1} holds, then there exists a $\delta > 0$ and $\rho_+ > r_+$, depending only on $a$, $l$, $M$ and $\alpha$, such that for all frequencies $|\omega-\omega_+|^2 \le \delta^2 \mathcal{L}^\star$ and all $r \in [r_+, \rho_+]$, we have
\begin{eqnarray} 
\tilde{V} &\ge& C \ \Delta_ - \left[ 1 + \mathcal{L} \right] \label{es:tvp} \, ,  \\
  \tilde{V}^\prime &\ge& C\  \Delta_-  \left[ 1 + \mathcal{L} \right]  \label{es:tvpp} \, ,  \\ 
\frac{\tilde{V}'}{\tilde{V}}&\ge& C \, .\label{es:tvpotvp}
%|V_{red}-\tilde{V}|&=\mathcal{O}(r-r_+) \nonumber
\end{eqnarray}
with $C > 0$ a constant depending only on $a$, $l$, $M$ and $\alpha$. 
\end{lemma}
\begin{proof} 
Note that
\begin{equation}
\tilde{V}=V_{red}+2(\omega-\omega_+)(\omega_r-\omega_+)-(\omega_r-\omega_+)^2 \, .  \label{eq:tvmvr} 
\end{equation}
To the first term we apply (\ref{pol}). Since for $m=0$ the estimates (\ref{eq:global})-(\ref{es:tvpotvp}) then follow immediately, we restrict to $m\neq 0$. The second term of (\ref{eq:tvmvr}) satisfies 
\[
| 2(\omega-\omega_+)(\omega_r-\omega_+)| \leq C_{M,l,a} \cdot \frac{\Delta_-}{(r^2+a^2)^2} \cdot \delta \cdot \mathcal{L}^\star \, .
\]
Returning to the first term $V_{red}$ and (\ref{pol}), we have from the properties of the $\lambda_{m\ell}$ collected in Lemma \ref{spheroidal}:
\begin{align}
\left( \mathcal{L}-2ma\Xi \omega \right)=\left( \mathcal{L}-2ma\Xi \omega_+ \right)+2ma\Xi(\omega- \omega_+) \nonumber \\
\ge \left(1-C_{M,l,a}\delta\right) \mathcal{L} - 2ma\Xi \omega_+ \nonumber \, ,
\end{align}
where we used Cauchy-Schwarz. Combining the two estimates and applying Corollary \ref{corhelp} we have
\begin{align}
\frac{\Delta_-}{(r^2+a^2)^2}\left( \mathcal{L}-2ma\Xi \omega \right) + 2(\omega-\omega_+)(\omega_r-\omega_+) \nonumber \\
\ge  \Bigg( \frac{1-C_{M,l,a} \delta}{1+\frac{\delta}{c_{M,l,a}}}  \left[ \Xi^2|m|(|m|+1) + \frac{m^2 a^4 \Xi^2}{\left(r_+^2 + a^2\right)^2} \right]  -2ma\Xi \omega_+ \Bigg)\frac{\Delta_-}{(r^2+a^2)^2} \, . \nonumber
\end{align}
Let us write $1-\tilde{\delta}$ for factor in front of the bracket above and finally subtract the last term of (\ref{eq:tvmvr}) from this expression. For $\tilde{\delta}=0$ we would have
\begin{align}
\left[\Xi^2 m^2 -2ma\Xi \omega_+ + \frac{m^2 a^4 \Xi^2}{\left(r_+^2 + a^2\right)^2}\right]\frac{\Delta_-}{(r^2+a^2)^2}-\left(\omega_r-\omega_+\right)^2 \nonumber \\
\geq \frac{\Delta_- \Xi^2 m^2}{\left(r^2+a^2\right)^2\left(r_+^2+a^2\right)^2} \left[ \left(r_+^2+a^2\right)^2 -2a^2 \left(r_+^2+a^2\right)+a^4 - \frac{a^2 \left(r^2-r_+^2\right)^2}{\Delta_-}\right]
\nonumber \\
\ge \frac{\Delta_- \Xi^2 m^2}{\left(r^2+a^2\right)^2\left(r_+^2+a^2\right)^2} \left[ r_+^4 - a^2 l^2 \right] \, ,\nonumber
\end{align}
where we used the estimate
\begin{align}
\frac{\left(r-r_+\right)\left(r+r_+\right)\left(r^2-r_+^2\right)}{\Delta_-} \leq l^2
\end{align}
in the last step, which in turn is a consequence of the identity
\begin{align} \label{eq:hepo}
\Delta_- = l^{-2} \left(r-r_+\right)\left(r^3 + r^2 r_+ + r \left( r_+^2 + a^2 +l^2\right) - a^2 l^2 r_+^{-1}\right) \, .
\end{align}
In view of the non-superradiance condition $r_+^2 > |a| l$, the last square bracket is bounded below by some $c_{M,l,a}$ and hence we can choose $\tilde{\delta}$ (hence $\delta$) sufficiently small to absorb the term $\tilde{\delta}  \left[ \Xi^2|m|(|m|+1) + \frac{m^2 a^4 \Xi^2}{\left(r_+^2 + a^2\right)^2} \right]$ ensuring that indeed (\ref{eq:global}) holds.

To establish  \eqref{es:tvp} we can redo the above computation with the only difference that near the horizon the term $\left(\omega_r - \omega_+\right)^2$ is of the order $m^2 \Delta_-^2 \leq \epsilon \Delta_- \mathcal{L}$ with the $\epsilon$ coming from the fact that $\Delta_-$ is small near the horizon. This leads to the better estimate
\begin{align} \label{eq:hie}
\frac{5}{6} \mathcal{L} - 2ma\Xi \omega_+ \geq c_{M,l,a} m^2
\end{align}
for sufficiently small $\delta$, therefore establishing (\ref{es:tvp}).

For \eqref{es:tvpp} and \eqref{es:tvpotvp}, recall the identity
\begin{align}
\tilde{V} = V_+ + V_\alpha + \frac{\Delta_-}{\left(r^2+a^2\right)^2} \left(\mathcal{L} - 2m a \Xi \omega\right) + 2 \left(\omega-\omega_+\right)\left(\omega_r-\omega_+\right) - \left(\omega_r - \omega_+ \right)^2 \nonumber \, ,
\end{align}
the fact that $\omega_r - \omega_+ \sim ma\Xi \Delta_-$ near the horizon and that $m^2 \lesssim \mathcal{L}$. From this we see that
\begin{align} \label{one1}
\tilde{V} \leq C_{M,l,a} \Delta_- \left(1 + \left((1+\tilde{\delta})\mathcal{L} - 2m a \Xi \omega_+ \right)\right) \leq C_{M,l,a} \Delta_- \left(1+\mathcal{L}\right)
\end{align}
holds for sufficiently small $\delta$ and $\rho_+$ sufficiently close to the horizon.
Similarly,
\begin{align} \label{two2}
\tilde{V}^\prime &\geq c_{M,l,a}\ \Delta_- \left(1 + \left(\mathcal{L} - 2m a \Xi \omega_+ \right)\right) - C_{M,l,a} \Delta_- \left(|m| \delta \sqrt{\mathcal{L}}  + m^2 \Delta_-\right) \nonumber \\
& \geq c_{M,l,a}\ \Delta_- \left(1 + \left(\left(1-C_{M,l,a}\delta\right)\mathcal{L} - 2m a \Xi \omega_+ \right)\right)
\nonumber \\
& \geq c_{M,l,a}\ \Delta_- \left(1 + c_{M,l,a} \mathcal{L}  \right) \, .
\end{align}
The last step following from (\ref{eq:hie}).
This establishes (\ref{es:tvpp}) and combining (\ref{one1}) and (\ref{two2}) yields also (\ref{es:tvpotvp}).
\end{proof}
\subsection{ $(\ell,m,\omega)$ low-frequency localisation} \label{sec:freqloc}
Let $\zeta\left(z\right)$ be a smooth cut-off function which is equal to $0$ for $z \geq 2$ and equal to $1$ for $z \leq 1$. Let moreover $\tilde{\zeta}\left(z\right)$ be an even cut-off function being equal to $1$ for $[-1/2,1/2]$ and equal to zero for $|z|\geq 1$.
We will cut off $\widehat{\psi^{\tau}}$ both in the $\omega$ variable and in the $(\ell,m)$ variables. Let us define, for fixed $L>1$, $\delta>0$ and $\mathcal{L}^\star=\max\left(\mathcal{L},\Xi^2\right)$ as defined in (\ref{regime1}), the following kernels:
\begin{align} \label{kernel}
\mathcal{K}_{\leq L, \omega \not\approx \omega_+} \left(t-s,\theta, \underline{\theta} ,\tilde{\phi},\tilde{\underline{\phi}}\right) 
= \int_{-\infty}^\infty d\omega \sum_{m,\ell}  \ \zeta\left(\frac{a^2 \omega^2 + \lambda_{m \ell}\left(\omega\right)}{L}\right)\cdot  \\ 
 \left[ 1 - \tilde{\zeta} \left(\left(\omega - \omega_+\right)\frac{1}{\delta \sqrt{\mathcal{L}^\star}} \right) \right]  e^{-i \omega \left(t-s\right)}  \cdot S_{m\ell} \left(a\omega,\cos \underline{\theta}\right)S_{m\ell} \left(a\omega,\cos \theta\right) e^{im \left(\tilde{\phi}-\tilde{\underline{\phi}}\right)} \, , \nonumber
\end{align}
\begin{align}
\mathcal{K}_{\leq L, \omega \approx \omega_+} \left(t-s,\theta, \underline{\theta} ,\tilde{\phi},\tilde{\underline{\phi}}\right) 
= \int_{-\infty}^\infty d\omega \sum_{m,\ell}  \ \zeta\left(\frac{a^2 \omega^2 + \lambda_{m \ell}\left(\omega\right)}{L}\right)  \cdot 
 \\  \tilde{\zeta} \left(\left(\omega - \omega_+\right) \frac{1}{\delta \sqrt{\mathcal{L}^\star}}\right) e^{-i \omega \left(t-s\right)}  \cdot S_{m\ell} \left(a\omega,\cos \underline{\theta}\right)S_{m\ell} \left(a\omega,\cos \theta\right) e^{im \left(\tilde{\phi}-\tilde{\underline{\phi}}\right)} \, , \nonumber
\end{align}
and the associated projection operators
\begin{align}
\mathcal{P}_{\leq L, \omega \not\approx \omega_+} \left(\varphi\right) \left(t,r,\theta,\tilde{\phi}\right) = \nonumber \\ 
\int_{-\infty}^\infty ds \int_{\mathbb{S}^2} \sin \theta d\underline{\theta} d\tilde{\underline{\phi}} \ \varphi \left(s,r,\underline{\theta},\underline{\tilde{\phi}}\right) \mathcal{K}_{\leq L, \omega \not\approx \omega_+} \left(t-s,\theta,\underline{\theta},\tilde{\phi},\underline{\tilde{\phi}}\right) \, , \nonumber
\end{align}
for $\varphi \left(s,r,\underline{\theta},\tilde{\underline{\phi}}\right)$ a smooth compactly supported (in $s$) function. Similarly for $\mathcal{P}_{\leq L, \omega \approx \omega_+} \left(\varphi\right) \left(t,r,\theta,\tilde{\phi}\right)$.  In applications, $\varphi$ will be the cut-off scalar field $\psi^\tau$ or derivatives thereof.  We finally define $\psi_{\flat}^{\tau}$, $\psi_{\sharp}^{\tau}$ as follows:
\begin{align}
\psi_{\flat}^{\tau} &= \psi_{\flat, \omega \not\approx \omega_+}^{\tau} +  \psi_{\flat, \omega \approx \omega_+}^{\tau} =  \mathcal{P}_{\leq L, \omega \not\approx \omega_+}(\psi^{\tau}) +  \mathcal{P}_{\leq L, \omega \approx \omega_+}(\psi^{\tau})  \, , \nonumber \\
\psi_{\sharp}^{\tau} &= \psi^{\tau}-\psi_{\flat}^{\tau} \nonumber \, .
\end{align}
Naturally, one has $\psi_{\flat}^{\tau}$ and $\psi_{\sharp}^{\tau}$ in $CH^k_{AdS}$ provided that $\psi$ is in $CH^k_{AdS}$, for all $k \ge 1$.

%
%
%
\begin{comment}
We have, for $a\neq 0$, 
\begin{align}
\| \mathcal{P}_{\leq L} \left(\varphi\right) \left(t,r,\theta,\tilde{\phi}\right)\|_{L^\infty}  \nonumber \\ \leq \| \varphi \|_{L^\infty} \cdot \sup_{t,\theta,\phi} \int_{-\infty}^\infty ds \int_{\mathbb{S}^2} \sqrt{\slashed{g}} d\underline{\theta} d\underline{\tilde{\phi}} \Big| \mathcal{K}_{\leq L} \left(t-s,\theta,\underline{\theta},\tilde{\phi},\underline{\tilde{\phi}}\right) \Big| 
\lesssim C \left(L\right) \| \varphi \|_{L^\infty} \nonumber \, 
\end{align}
since the compact support in $\omega$ makes the kernel a Schwartz function in $s$. However, a-priori we have no control on the constant $C\left(L\right)$.\footnote{This is because the Schwartz property of the kernel in the variable $s$ relies on estimates for $\omega$-derivatives of the $S_{ml}$ for which we only have compactness and smoothness arguments available at the moment.} 
\begin{remark}
If $a=0$, we do not cut-off the support of $\omega$ at all and hence $\mathcal{K}_{\leq L}$ is not Schwartz. Since in the $a=0$ case one does not need to work with $\psi^\tau$ and the Fourier transform at all (and in fact the proof is easier, see ...), there is no loss in assuming $a \neq 0$. Moreover, once can easily incorporate the $a=0$ case at this point by simply inserting an additional cut-off function $\zeta\left(\frac{\omega^2}{L}\right)$ into (\ref{kernel}).
\end{remark}
\end{comment}
%
%
%
%
%
%
%
%
%
%
\section{Microlocal current templates} \label{se:mct}
In this section, we shall define currents at the level of the separated Fourier components $u^{(a\omega)}_{m\ell} \left(r\right)$ introduced in Section \ref{se:sepv}. These  currents are clearly well-defined for $\psi_{\flat}^{\tau}$ since in this case the $u^{(a\omega)}_{m\ell} \left(r\right)$ are all smooth. 
%In section .., we shall prove that $\psi_{\flat}$ (and its time derivatives) are $L^2$ in time, from which it follows that these currents are also well-defined for $\psi_{\flat}$.

To avoid overloading the notation, we will employ in this section the short hand notation $u=u^{(a\omega)}_{m\ell} \left(r\right)$, $V=V^{(a\omega)}_{m\ell} \left(r\right)$ and $H=H^{(a\omega)}_{m\ell} \left(r\right)$.

\subsection{Microlocal energy and red-shift currents}
The microlocal current associated to the Killing fields $T=\partial_t$ and $K= \partial_t + \frac{a\Xi}{r_+^2 + a^2} \partial_\phi$ is given by
\begin{align} \label{Tcurr}
Q_T[u]=\omega Im \left( u' \bar{u} \right) \ \ \ \ \textrm{and}  \ \  \ \ Q_K \left[u\right] = \left(\omega - \omega_+\right) Im \left( u' \bar{u} \right) 
\end{align}
respectively. We have $Q_T'=\omega Im \left(H \overline{u}\right)$ and $Q_K^\prime= \left(\omega - \omega_+\right) Im \left( H \bar{u} \right)$. Hence for $H=0$,  the quantities $Q_T$ and $Q_K$ are constant in $r^\star$.

Recall the definitions of (\ref{vred}).
For any real function $z$, we also introduce the microlocal red-shift current:
\begin{align} \label{redshift}
Q_{red}^z = z\left[\left| u' +i\left(\omega - \omega_+\right)u \right|^2 - V_{red} \left|u \right|^2 \right] \, .
\end{align}
The current $Q_{red}^z$ satisfies
\begin{align} 
Q_{red}'=& z^\prime \left| u' +i\left(\omega - \frac{a \Xi m}{r_+^2+a^2}\right)u \right|^2 - \left(z V_{red} \right)^\prime|u|^2 \nonumber \\ &\hbox{}-z(|u|^2)'\left(V_{red}-\tilde{V}\right)+ 2z Re \left(\left(u^\prime+i\left(\omega - \frac{a \Xi m}{r_+^2+a^2}\right)u \right) \overline{H}\right) \, .
\end{align}
\begin{remark}
The microlocal redshift current will only be used in the proof of Theorem \ref{theorem2}. The proof of Theorem \ref{theorem1} will exploit the redshift directly in physical space in conjunction with the pigeonhole principle after some version of integrated local energy decay has been proven.
\end{remark}
\subsection{Microlocal Morawetz currents}
We finally turn to the Morawetz currents.
Let $f$ and $g$ be bounded functions of $r^\star$ only. Using the notation of \cite{DafRodsmalla, DafRodlargea}, we define 
\begin{align} \label{fcurrent}
Q^f_0 = f \left[ | u^\prime |^2 + \left(\omega^2 - V\right)|u|^2\right] + f^\prime Re \left(u^\prime \overline{u}\right) - \frac{1}{2} f^{\prime \prime} |u|^2 \, ,
\end{align}
\begin{align} \label{mlhc}
Q^h_1 = h Re \left(u^\prime \overline{u}\right) - \frac{1}{2} h^\prime |u|^2 \, ,
\end{align}
\begin{align}
Q^g_2 = g \left[ | u^\prime |^2 + \left(\omega^2 - V\right)|u|^2\right] \, ,
\end{align}
which satisfy
\begin{align} \label{fbulk}
\left(Q^f_0\right)^\prime = 2 f^\prime | u^\prime|^2 - \left( f V^\prime +  \frac{1}{2} f^{\prime \prime \prime} \right) |u|^2 + 2 f Re \left(u^\prime \overline{H}\right) + f^\prime Re \left(u \overline{H}\right) \, ,
\end{align}
\begin{align}
\left(Q^h_1\right)^\prime = h \left(|u^\prime|^2 + \left(V-\omega^2\right)|u|^2 \right) - \frac{1}{2}h^{\prime \prime} |u|^2 + h Re \left(u\overline{H}\right) \, ,
\end{align}
\begin{align}
\left(Q^g_2\right)^\prime =  g^\prime \left[ | u^\prime |^2 + \left(\omega^2 - V\right)|u|^2\right] - gV^\prime |u|^2 +   2g Re \left(u^\prime \overline{H}\right) \, .
\end{align} 
In the following, the combination
\begin{align}
Q^{g,h} = Q_2^g + Q_1^{h} 
\end{align}
 will be particularly useful. For this choice
 \begin{align} \label{mixedM}
\left( {Q}^{g,h} \right)^\prime = |u^\prime|^2 \left(g^\prime + h \right) + \omega^2 \left(g^\prime - h \right) |u|^2 \nonumber \\ + |u|^2 \left(-g^\prime V -g V^\prime +V h - \frac{1}{2} h^{\prime \prime} \right) 
+ 2g Re \left(u^\prime \overline{H}\right) + h Re \left(u\overline{H}\right) \, .
 \end{align}

\section{Microlocal estimates for $\psi_{\flat, \omega \approx \omega_+}^{\tau}$} \label{phiflatt2}

Recall the parameter $\delta$, which determines the relation between the size of $\omega-\omega_+$ and $\mathcal{L}$, cf.~Section \ref{sec:freqloc}. In this section 
% we prove an integrated decay estimate for frequencies very close to $\omega = \omega_+$. More precisely, 
we are going to determine a $\delta$ (small) depending only on the parameters, such that for frequencies satisfying $|\omega - \omega_+ | < \delta \sqrt{\mathcal{L}^\star}$ and $\mathcal{L}< 2L$, we have an integrated decay estimate. Once this constant $\delta$ has been determined it will be fixed for the rest of the paper.

\begin{remark}
Similar estimates to the one's presented below have appeared before in a different context in \cite{DafRodsmalla, DafRodlargea, Aretakis} for the ``almost stationary regime" ($|\omega - \omega_+| \leq \delta$). Here we adapt the construction of \cite{Aretakis} because we want to avoid using the redshift at this point (cf.~the remarks in Section \ref{sec:hozimprove}). While the multiplier constructed below will cover the larger ``angular dominated" range $|\omega - \omega_+ | < \delta \sqrt{\mathcal{L}^\star}$, the notation $\psi_{\flat, \omega \approx \omega_+}^{\tau}$ is still reminiscent of the original setting.
\end{remark}

From Lemma \ref{lem:Vtilde}, we first note that we can pick $\delta$ in $|\omega - \omega_+| \leq \delta \sqrt{\mathcal{L}^\star}$ sufficiently small such that we have the global estimate
\begin{align}
\tilde{V} + \frac{1}{4l^2} \frac{\Delta_-}{r^2+a^2}  \geq c_{M,l,a} \frac{\Delta_-}{\left(r^2+a^2\right)^2} \left(r^2 + m^2\right) \, ,
\end{align}
where the second term on the left hand side is not needed if $\alpha<1$. For the mass range $\alpha\geq1$ this additional term is gained from the derivative term via a Hardy inequality:
\begin{lemma} \label{lem:hardyabsorb}
For any $r_{cut}\geq r_+$, we have for a smooth function $u$ with $u r^{\frac{1}{2}} = o(1)$ at infinity 
\begin{align} \label{helpest}
\frac{1}{4 l^2} \int_{r=r_{cut}}^\infty dr^\star \frac{\Delta_-}{\left(r^2+a^2\right)} |u|^2 \leq  \int_{r=r_{cut}}^\infty dr^\star |u^\prime |^2 \, .
\end{align}
\end{lemma}
\begin{proof}
Note that the weight on the left is the $r^\star$ derivative of $r$. Integrating by parts and Cauchy-Schwarz yields
\begin{align}
\int_{r=r_{cut}}^\infty \left(r\right)^\prime |u|^2  \ dr^\star \leq 4  \int_{r=r_{cut}}^\infty dr^\star \frac{r^2 +a^2} {\Delta_-} \left(r-r_{cut}\right)^2 |u^\prime |^2 \, .
\end{align}

The Lemma now follows by estimating $r_{cut} \leq r_+$ and using (\ref{eq:hepo}).
\end{proof}
Next recall the multiplier (\ref{mixedM}) written in the form
 \begin{align} \label{niceform}
\left( {Q}^{g,h} \right)^\prime = |u^\prime|^2 \left(g^\prime + h \right) + g^\prime \left(\omega-\omega_+\right)^2 |u|^2 \nonumber \\ + |u|^2 \left[\left(-g\tilde{V} \right)^\prime - \frac{1}{2} h^{\prime \prime} + h \left(V - \omega^2 \right) \right]
+ 2g Re \left(u^\prime \overline{H}\right) + h Re \left(u\overline{H}\right) \, .
 \end{align}
If $\omega=\omega_+$, we could simply choose $g=0$ and $h=1$ to obtain an integrated decay estimate in view of (\ref{vred}), (\ref{eq:global}) and (\ref{helpest}).\footnote{A variant of this argument will be used in Section \ref{noperiodic} to exclude time-periodic solutions.} However, since we only have $|\omega-\omega_+| \leq \delta \sqrt{\mathcal{L}^\star}$, we need to cut-off $h$ and control the errors using the function $g$.

Let $r_{+} < r_1 < r_2 < r_3$ (with corresponding $r^\star_{1,2,3}$) be $r$-values close to the horizon such that $r^\star_3 < 0$ and moreover (cf.~Lemma \ref{lem:Vtilde})
\begin{align} \label{propVr2r3}
\tilde{V} > 0 \, ,  \ \ \  \frac{\tilde{V}^\prime}{\tilde{V}} \geq - \frac{2}{r^\star_2}  \textrm{ \ \ and \ \  } \frac{\tilde{V}^\prime}{\tilde{V}} \geq \frac{1}{r} \frac{4 \Delta_-}{r^2+a^2} \ \  \textrm{hold in $r_+ < r\leq r_2$} \, .
\end{align}
From now on we regard $r_2$ and $r_3$ as being fixed, depending only on the above conditions and hence depending only on $M$, $l$ and $a$ and $\alpha$ in view of Lemma \ref{lem:Vtilde}. The value of $r_1$ will be chosen below.
Define $\underline{g}$ as follows:
\begin{large}
\begin{equation*} \underline{g} \left(r^\star\right) = \left\{
\begin{array}{rl} 
\frac{1}{r^\star_1 \tilde{V} \left(r=r_1\right)} & \text{if } r < r_1,\\ 
\frac{1}{r^\star \cdot \tilde{V}\left(r\right)}  & \text{if } r_1 \leq r < r_2,\\ 
\frac{r^\star - r^\star_3}{\left(r^\star_2 - r^\star_3\right) r^\star_2 \tilde{V} \left(r=r_2\right)} & \text{if } r_2 \leq r < r_3,\\ 
0 & \text{if } r \geq r_3
\end{array} \right. \end{equation*}
\end{large}
and hence
\begin{large}
\begin{equation*} \underline{g}^\prime \left(r^\star\right) = \left\{
\begin{array}{rl} 
0 & \text{if } r < r_1,\\ 
- \frac{1}{r^\star \cdot \tilde{V}\left(r\right)} \left[\frac{\tilde{V}^\prime}{\tilde{V}} + \frac{1}{r^\star}\right]  & \text{if } r_1 \leq r < r_2,\\ 
\frac{1}{\left(r^\star_2 - r^\star_3\right) r^\star_2\tilde{V} \left(r=r_2\right)} & \text{if } r_2 \leq r < r_3,\\ 
0 & \text{if } r \geq r_3.
\end{array} \right. \end{equation*}
\end{large}
Note that $\underline{g}$ is $C^0$ and piecewise $C^1$ (hence the regularity is sufficient for the multiplier identity (\ref{niceform}) to hold), everywhere non-positive and monotonically increasing. 
Define $\underline{h}$ as
\begin{large}
\begin{equation*} \underline{h} \left(r^\star\right) = \left\{
\begin{array}{rl} 
0 & \text{if } r < r_1, \\
\tilde{h}\left(r^\star\right) & \text{if }  r_1 \leq r < r_2, \\
\frac{4}{\left(r^\star_2 - r^\star_3\right) r^\star_2\tilde{V} \left(r=r_2\right)} & \text{if } r \geq r_2.
\end{array} \right. \end{equation*}
\end{large}
where $\tilde{h} \left(r^\star\right)$ is a smooth monotone non-negative function interpolating smoothly between the two constant values at $r=r_1$ and $r=r_2$. 

Using in particular that $\tilde{V}(r_2)$ can be bounded from below independently of $m$, we can choose $r_1$ sufficiently close to the horizon (depending only on $M$, $l$ and $a$, in particular not on $m$) so that there exists a $\tilde{h} \left(r^\star\right)$ satisfying $|\tilde{h}^{\prime \prime} \left(r^\star\right)| < \frac{1}{\left(r^\star\right)^2}$  in $r_1 \leq r < r_2$. This fixes our $\tilde{h}\left(r^\star\right)$ and our $r_1$. 

Finally, with Lemma \ref{lem:Vtilde} in mind, we choose $\delta$ sufficiently small (again depending only on $M$, $l$ and $a$) such that the following holds for frequencies 
 $|\omega - \omega_+| \leq \delta \sqrt{\mathcal{L}^\star}$:
\begin{align}
V - \omega^2 &= \tilde{V} - \left(\omega-\omega_+\right)^2 \geq \frac{\tilde{V}}{2}  \textrm{ \ \ \ in $r_2 \leq r$} \, , \nonumber \\
V - \omega^2 &= \tilde{V} - \left(\omega-\omega_+\right)^2 \geq 0  \textrm{ \ \ \ \ in $r_1 \leq r \leq r_2$} \, .\nonumber
\end{align}
Note that with this choice we have in $r_1 \leq r \leq r_2$ 
\begin{align}
\underline{h}\left(V-\omega^2\right) + \left(-\underline{g} \tilde{V}\right)^\prime - \frac{1}{2}\underline{h}^{\prime \prime} \geq \frac{1}{\left(r^\star\right)^2} - \frac{1}{2 \left(r^\star\right)^2} = \frac{1}{2 \left(r^\star\right)^2} \, ,
\end{align}
while in $r_2 \leq r \leq r_3$ we have
\begin{align}
\underline{h}\left(V-\omega^2\right) + \left(-\underline{g} \tilde{V}\right)^\prime - \frac{1}{2}\underline{h}^{\prime \prime} = \underline{h} \left(\tilde{V} - \left(\omega - \omega_+\right)^2 \right) - \underline{g} \tilde{V}^\prime - \underline{g}^\prime \tilde{V} + 0 \nonumber \\
\geq \left( \frac{\underline{h}}{2} - \underline{g}^\prime \right) \tilde{V} \geq \frac{1}{L} c_{M,l,a}  \frac{\Delta_-}{\left(r^2+a^2\right)^2} \left(m^2 + r^2\right) \, ,
\end{align}
with the $L^{-1}$ appearing from the size of $\underline{h}$.\footnote{Note that near the horizon one has the estimate $| \tilde{V}| \lesssim \Delta_- L$.}
Finally, for $r_3 \leq r < \infty$, $\underline{h}$ is constant. We can apply Lemma \ref{lem:hardyabsorb} with $r_{cut}=r_3$ to the derivative term $\underline{h}|u^\prime|^2$ to gain the zeroth order term which makes (\ref{eq:global}) available. With this, we have, for sufficiently small $\delta$,
\begin{align}
\int_{r=r_3}^{r=\infty} dr^\star \  \underline{h} \left[ \ |u^\prime|^2 + \left(V-\omega^2\right) |u|^2 \right]  \geq \frac{c_{M,l,a}}{L}  \int_{r_3}^\infty dr^\star \   |u|^2 \frac{\Delta_-}{\left(r^2+a^2\right)^2} \left(m^2 + r^2\right) .\nonumber
\end{align} 
Note that this borrowing via Lemma \ref{lem:hardyabsorb} is not necessary if $\alpha<1$.
We summarize the construction in the following
\begin{lemma} \label{prop:im1}
Under any of the conditions in Theorem \ref{theorem1}, there exists an $\delta$ depending only on $M$, $l$ and $a$ such that the following holds for frequencies $|\omega - \omega_+| \leq \delta \sqrt{\mathcal{L}^\star}$:
\begin{align} \label{im1est}
\frac{1}{L} \cdot c_{M,l,a} \Bigg[ \int_{r^\star_2}^\infty |u|^2  \frac{\Delta_-}{\left(r^2+a^2\right)^2} \left(m^2 + r^2\right) dr^\star + \int_{r^\star_1}^{r^\star_2} \frac{1}{\left(r^\star\right)^2} |u|^2 dr^\star \nonumber \\
+ \int_{R^\star_{-\infty}}^{r^\star_1} |u|^2 \frac{\Delta_-}{\left(r^2+a^2\right)^2} \left(m^2 + r^2\right) dr^\star + \boxed{\int_{r^\star_1}^{\infty} |u^\prime|^2 dr^\star} \  \Bigg] \leq -\mathcal{Q}^{\underline{g},\underline{h}} \left(R^\star_{-\infty}\right) \nonumber \\
- \int_{R^\star_{-\infty}}^\infty \Big[ 2\underline{g} Re \left(u^\prime \overline{H}\right) + \underline{h} Re \left(u\overline{H}\right) \Big] dr^\star
\end{align}
\end{lemma}
Here we have used that $\mathcal{Q}^{\underline{g},\underline{h}} \left(r^\star_{\infty}\right)=0$ in view of the decay of $u$ and $u^\prime$, cf.~(\ref{udecrates}). The factor of $L^{-1}$ appears since the functions $\underline{g}$ and $\underline{h}$ scale with $L^{-1}$. Note also that for $\alpha=2$ one first obtains this estimate without the boxed term, since it was used entirely when applying the Hardy inequality of Lemma \ref{lem:hardyabsorb}. However, one can simply repeat the estimate once more to regain control of the derivative term. 

The $\underline{h}$-error-term in the last line of (\ref{im1est}) is supported only for $r\geq r_1$ and easily absorbed on the left using Cauchy's inequality
\begin{align} \label{ch}
| \underline{h} Re \left(u\overline{H}\right) | \leq \frac{\epsilon}{L} \frac{\Delta_-}{\left(r^2+a^2\right)^2} |u|^2 + L \frac{|\underline{h}_{max}|^2}{\epsilon} \frac{\left(r^2+a^2\right)^2}{\Delta_-} |H|^2 \, .
\end{align}
As we will see later, the last term in (\ref{ch}) can be controlled from the boundedness statement.
The $\underline{g}$-error in (\ref{im1est}) is more problematic, since we do not have control over the derivative $u^\prime$ in the region $r \leq r_1$ on the left hand side. We need an auxiliary estimate. Choose $g= -r^{-2}$ and $h=0$ in (\ref{niceform}). Using again the properties of $\tilde{V}$ (the last of (\ref{propVr2r3})),  we obtain an estimate of the form
\begin{align}
c_{M,l,a} \int_{R^\star_{-\infty}}^{r^\star_2} \left( |u^\prime|^2 + \left(1+ \mathcal{L}\right) |u|^2 \right)  \frac{\Delta_-}{\left(r^2+a^2\right)^2} dr^\star \leq  -\mathcal{Q}^{-r^{-2},0} \left(R^\star_{-\infty}\right)   \nonumber \\
\leq C_{M,l,a} \int_{r^\star_{2}}^{\infty} |u|^2 \frac{\Delta_-^2}{\left(r^2+a^2\right)^4} \frac{1}{r} \left(L + r^2\right) dr^\star 
+ C_{M,l,a} \int_{R^\star_{-\infty}}^\infty  \frac{1}{r^2} | Re \left(u^\prime \overline{H}\right)| dr^\star \nonumber
\end{align}
Combining this estimate with the estimate of Proposition \ref{prop:im1}
 %and adding an angular term to the left hand side taking into account that $\lambda_{m \ell} + a^2\omega^2 \leq L$, 
 we finally find 

\begin{proposition}Under the condition of Lemma \ref{prop:im1}, we have
\begin{align} \label{mainnonstat}
 \Bigg[ \int_{R^\star_{-\infty}}^\infty  \frac{\Delta_-}{\left(r^2+a^2\right)^2} \left[  |u^\prime |^2 + |u|^2 \left(m^2 + r^2 + \lambda_{m\ell} + a^2 \omega^2 \right)   \right] dr^\star  \lesssim \\
-L^2 \cdot \mathcal{Q}^{\underline{g},\underline{h}} \left(R^\star_{-\infty}\right) - L^2 \cdot \mathcal{Q}^{-r^{-2},0} \left(R^\star_{-\infty}\right)
+ L^3  \int_{R^\star_{-\infty}}^\infty dr^\star \frac{\left(r^2+a^2\right)^2}{\Delta_-} |H|^2 \, , \nonumber
\end{align}
where $u=\sqrt{r^2+a^2} \psi^{(a\omega)}_{m \ell}$ denotes the Fourier-separated components of $\psi^\tau_{\flat, \omega \approx \omega_+}$.
\end{proposition}
%where we have also 
\begin{remark} The dependence on $L$ could possibly be removed using a more refined analysis. However, since we are going to lose exponentially in $L$ in the estimate for $\psi_{\flat, \omega \not\approx \omega_+}^{\tau}$, we content ourselves with this rather crude estimate.
\end{remark}
\section{Microlocal estimates for $\psi_{\flat, \omega \not\approx \omega_+}^{\tau}$} \label{phiflatt}
In this section, we shall prove an integrated decay estimate for $\psi_{\flat, \omega \not\approx \omega_+}^{\tau}$. Recall that this part of the solution satisfies in addition to $\mathcal{L}<2L$ the estimate
\[
|\omega-\omega_+| \geq \frac{\delta}{2} \sqrt{\mathcal{L}^\star} \, .
\] 
We will write
$\Psi$ for the separated Fourier-components of $\psi_{\flat, \omega \not\approx \omega_+}^{\tau}$ and as usual, $u = \sqrt{r^2+a^2} \Psi$.

To keep track of the various parameters involved in this section, we note already at this stage that we will, depending on the size of the mass $\alpha$, first choose numbers $n_2$ and $n_1$.  Finally, depending on $\delta$ (\emph{fixed} already in the previous section and depending only on the black hole parameters $M$, $l$, $a$), the constants $n_1, n_2$, the black hole parameters and the cut-off constant $L$, we will choose a large constant $k = C_{M,l,a,\alpha, n_1,n_2} \cdot \sqrt{L}$. 

\subsection{Preliminaries}
Key to our argument in this frequency range is the use of exponentially weighted multipliers, which we introduce now.

\subsubsection{Exponentially weighted multipliers}
Let $x=\arctan \left(\frac{r}{l}\right)$. Note $c_{M,l,a} l  \leq r^2  \cdot x_r = l \left[1- \frac{1}{1+\frac{r^2}{l^2}}\right] \leq l $ and hence that the expression $r^2 x_r$ is essentially like a constant bounded away from zero and infinity. We define 
\begin{align}
f_n = \frac{1}{r^n} \exp \left(-k\arctan \left(\frac{r}{l}\right)\right) = \frac{1}{r^n} e^{-kx} \, ,
\end{align}
\begin{align}
f_n^\prime = \left(\frac{-n}{r^{n+1}} - \frac{k \ x_r}{r^n}\right)\frac{\Delta_-}{r^2+a^2} e^{-kx} =  \left(\frac{-n}{r} - k \ x_r\right) \frac{\Delta_-}{r^2+a^2} f_n = - \sigma f_n \, , \nonumber
\end{align}
where
\begin{align}
\sigma = \left(\frac{n}{r} +  k \ \frac{(x_r r^2)}{r^2} \right) \frac{\Delta_-}{r^2+a^2} \, .
\end{align}
With this we compute
\begin{align}
 f_n^{\prime \prime} = \left(\sigma^2-\sigma^\prime \right) f_n  \ \ \ \ \ \textrm{and} \ \ \ \ f_n^{\prime \prime \prime} = -\left(\sigma^3 - 3 \sigma \sigma^\prime + \sigma^{\prime \prime}\right) f_n
\end{align}
Differentiating $\sigma$, we find 
\[
| \sigma^\prime| \leq \frac{\Delta_-}{r^2+a^2} \left[ \frac{n}{l^2} + C_{M,l,a}  \left(\frac{n}{r^2} +  k \ \frac{(x_r r^2)}{r} \right) \right] \,
\]
\[
|\sigma^{\prime \prime}| \leq \frac{2}{r} \left(\frac{r^2}{l^2} + \frac{M}{r} \right) | \sigma^\prime| + C_{M,l,a} \left(\frac{\Delta_-}{r^2+a^2}\right)^2  \left[ \frac{n}{r^3} +  k \ \frac{(x_r r^2)}{r^2} \right] \, .
\]
It follows that away from the horizon, in $r\geq 3M$, we have for sufficiently large $n$ and $k$ (depending only on $M, l, a$) the bound
\begin{align} \label{outside}
\sigma^3 - 3 \sigma \sigma^\prime + \sigma^{\prime \prime} \gtrsim r^3 \left[ n^3 + \frac{k n^2 (x_r r^2)}{r} + \frac{n k^2(x_r r^2)^2}{r^2} + \frac{k^3(x_r r^2)^3}{r^3} \right] .
\end{align}
To see this, note that this estimate holds for $\sigma^3$ and absorb the other terms using the estimates for $|\sigma^\prime|$ and $|\sigma^{\prime \prime}|$ above, recalling that $\frac{\Delta_-}{r^2+a^2} \sim r^2$ in this region.

For $r\leq 3M$ the expression on the left hand side of (\ref{outside}) is not positive everywhere. However, we can find a constant $C_{M,l,a}$ such that 
\begin{align} \label{inside}
\sigma^3 - 3 \sigma \sigma^\prime + \sigma^{\prime \prime} + C_{M,l,a} \cdot \mathcal{L}^\star \cdot \sigma \gtrsim \left( k^3\Delta_-^3 + k^2 \Delta_-^2 + k \cdot \mathcal{L}^\star \cdot \Delta_-\right) 
\end{align}
holds in $ r\leq 3M$, again provided $k$ is chosen sufficiently large.\footnote{The reason $\mathcal{L}^\star$ appears here will become clear in the following Lemma. The estimate also holds for $\mathcal{L}^\star=1$.} To see this, note that the estimate holds for $\sigma^3 + C_{M,l,a}  \mathcal{L}^\star \cdot \sigma$ and that we can make the terms $k^2 \Delta_-^2$ and $k \Delta_- \mathcal{L}^\star$ large by adding more of $C_{M,l,a} \cdot \mathcal{L}^\star \cdot \sigma$. Since $r$-weights are irrelevant in $r\leq 3M$, we have $|\sigma \sigma^\prime| + |\sigma^{\prime \prime}| \lesssim (k+n)^2 \Delta_-^2 + (k+n) \Delta_-$ and hence these terms can be absorbed by $\sigma^3$ and the $C_{M,l,a} \cdot \mathcal{L}^\star \cdot \sigma$-term. We apply the same reasoning in the following Lemma, which has obviously been tailored to be applicable to the derivative of the multiplier (\ref{fcurrent}) later.
\begin{lemma} \label{bulklemma}
If $f=f_n$ we have, for sufficiently large $n$ and $k$ (in particular $k \geq C_{M,l} \sqrt{L} \geq C_{M,l} \left(\lambda_{m \ell} + a^2\omega^2\right)$) and for all $r \geq r_+$ the estimate
\begin{align}
 - \left( f V^\prime +  \frac{1}{2} f^{\prime \prime \prime} \right)
 \geq - C_{M,l,a,\alpha} \cdot k \cdot \mathcal{L}^\star \cdot \left(\frac{\Delta_-}{r^2+a^2} \right) \frac{e^{-kx} }{r^n}\chi_{\{r \leq 3M\}}   \nonumber \\ +  c_{M,l,a, \alpha} \left[ k^2 \left(\frac{\Delta_-}{\left(r^2+a^2\right)^2} \right)^2 \frac{1}{r^{n-1}} + k^3 \left(\frac{\Delta_-}{\left(r^2+a^2\right)^2} \right)^3 \frac{1}{r^{n}} \right]  e^{-kx}  \nonumber \, .
\end{align}
\end{lemma}
\begin{remark}
Note that the first term on the right hand side is localized. 
\end{remark}
\begin{proof}
For $-f^{\prime \prime \prime}$ the estimate holds from (\ref{inside}) and (\ref{outside}). For the potential we have
\begin{align}
|f V^\prime| \leq  \frac{\Delta_-}{\left(r^2+a^2\right)^2} C_{M,l,a} \left( \left( |\alpha| +  1\right) r^3 +  {r}^{-1} \mathcal{L}^\star \right) \ \frac{e^{-kx}}{r^n} \textrm{ \ \  \ for all $r \geq r_+$.} \nonumber
\end{align}
Comparing with (\ref{outside}) and (\ref{inside}), it follows that the $\mathcal{L}^\star$-term can be absorbed by the $k \cdot\mathcal{L}^\star$ term near the horizon and by the good $k^3 \Delta_-^3$-term away from it for sufficiently large $n$ and $k$. The $\left(|\alpha|+1\right)$-term, on the other hand, is absorbed by the $n^3$-term in (\ref{outside}) for sufficiently large $n$.
\end{proof}

\subsubsection{An exponentially weighted Hardy inequality}
A second key tool is the use of Hardy inequalities with exponential weights.
\begin{lemma} \label{lem:hardy1}
Let $\beta \in [0,1]$ and $u\left(r^\star\right)$ be a smooth, bounded, complex valued function. Then, for positive numbers $k$ and $n\geq 1$, we have the Hardy inequality
\begin{align}
\int_{r=5M}^{\infty} \frac{e^{-kx}}{r^{n-1}}  \left(n+\frac{k \left(x_r r^2\right)}{r}\right) |u|^2 dr^\star \nonumber \\ \leq 16 l^4 \int_{r=4M}^{\infty} \frac{e^{-kx} }{r^{n+1}} \left(\frac{1-\beta}{n} + \beta \frac{r}{k \left(x_r r^2\right)} \right) |u^\prime|^2 dr^\star \nonumber \\ + C_{M,l,a} \int_{r=4M}^{r=5M} \frac{e^{-kx}}{r^{n-1}} \left(n+\frac{k \left(x_r r^2\right)}{r}\right) |u|^2  dr^\star \nonumber \, .
\end{align}
\end{lemma}
\begin{proof}
Let $\chi \left(r\right)$ be a cut-off function which is equal to $1$ for $r \geq 5M$ and zero for $r \leq 4M$. Let $4M^\star$ denote the $r^\star$-value at $r=4M$. We integrate by parts,
\begin{align}
\int_{4M^\star}^{R^\star_\infty} \left(-\frac{e^{-kx}}{r^n} \right)^\prime \chi |u|^2 dr^\star = \int_{4M^\star}^{5M^\star} \frac{e^{-kx}}{r^n} \chi^\prime |u|^2 dr^\star + 2 \int_{4M^\star}^{R^\star_\infty} \frac{e^{-kx}}{r^n} \chi \ Re \left(\bar{u} u^\prime\right) \nonumber \\
\leq \frac{1}{2} \int_{4M^\star}^{R^\star_\infty} \left(-\frac{e^{-kx}}{r^n} \right)^\prime \chi \ |u|^2 dr^\star + 
2 \int_{4M^\star}^{R^\star_\infty} \frac{e^{-kx}e^{-kx}}{r^{2n}  \left(-\frac{e^{-kx}}{r^n} \right)^\prime } \ \chi \ |u^\prime|^2 dr^\star
\nonumber \\
+\int_{4M^\star}^{5M^\star} \frac{e^{-kx}}{r^n} \chi^\prime |u|^2 dr^\star  \nonumber
\end{align}
observing that the boundary terms vanish in view of $\chi$ vanishing at $4M$ and $u$ being bounded near infinity. This estimate leads to
\begin{align}
\int_{4M^\star}^{R^\star_\infty} \left(-\frac{e^{-kx}}{r^n} \right)^\prime \chi \ |u|^2 dr^\star \nonumber \\\leq 4  \int_{4M^\star}^{R^\star_\infty} \frac{e^{-kx}e^{-kx}}{r^{2n}  \left(-\frac{e^{-kx}}{r^n} \right)^\prime } \ \chi \ |u^\prime|^2 dr^\star + 2 \int_{4M^\star}^{5M^\star} \frac{e^{-kx}}{r^n} \chi^\prime |u|^2 dr^\star \, .
\end{align}
Observe now that 
\begin{align}
\left(-\frac{e^{-kx}}{r^n} \right)^\prime = e^{-kx} \left( \frac{n}{r^{n+1}} + \frac{k \ r^2 \cdot x_r }{r^{n+2}} \right) \frac{\Delta_-}{r^2+a^2} \, .
\end{align}
Taking into account
\begin{align}
\frac{1}{n+ \frac{k \left(r^2 x_r\right)}{r}} \leq \frac{1-\beta}{n} + \beta \frac{r}{k \left(r^2 x_r\right)}  \ \ \ \textrm{and} \ \ \ \frac{r^2+a^2}{\Delta} \leq \frac{2 l^2}{r^2}  \textrm{  \ \ \ for $r\geq 4M$} \, ,
\end{align}
the inequality of the Lemma follows.
\end{proof}
\subsection{Estimating $|u^\prime|^2$ from $|u|^2$}
We turn to the estimates for the solution.
We first apply the current ${Q}^{g}_2$ with $g = -f_{n_2}$ (hence $g^\prime = \sigma f_{n_2}>0$). As in (\ref{niceform}), we see that the spacetime term can be written
\begin{align}
\left({Q}^{g}_2\right)^\prime = |u^\prime|^2 \ g^\prime + |u|^2 \left[g^\prime \left(\omega-\omega_+\right)^2 + \left(-g\tilde{V} \right)^\prime \right]
+ 2g Re \left(u^\prime \overline{H}\right) \, .
\end{align}
We next study the behavior of the zeroth order term
\begin{align}
\mathcal{V} =  \frac{1}{2} g^\prime \left(\omega-\omega_+\right)^2 + \left(-g\tilde{V} \right)^\prime\, .
\end{align}
Note that we are keeping the globally good (=positive) term $ \frac{1}{2} g^\prime \left(\omega-\omega_+\right)^2$ on the left.
\begin{lemma}

For frequencies $|\omega-\omega_+| \geq \frac{\delta}{2} \sqrt{\mathcal{L}^\star}$ we can find a $\tilde{\rho}_+$ such that
\[
 \frac{1}{2} g^\prime \left(\omega-\omega_+\right)^2 + \left(-g\tilde{V} \right)^\prime \geq 0
\]
 holds for all $r_+ \leq r \leq \tilde{\rho}_+$.
\end{lemma}
\begin{proof}
The first term is positive and at least of size $\sim k \Delta_- \cdot \delta^2 \cdot \mathcal{L}^\star$ near the horizon. For the second term we observe (recalling (\ref{one1}), for instance) that $-g^\prime \tilde{V} \sim k \Delta_-^2 \left( \mathcal{L} + 1\right)$ while $-g \tilde{V}^\prime$ has a positive sign near the horizon. We conclude that choosing $\tilde{\rho}_+$ such that $\Delta_- \ll \delta$ holds in $r_+\leq r\leq \tilde{\rho}_+$ (with $\ll$ depending only on $M$,$l$ and $a$) we can achieve positivity.
\end{proof}

Fixing now the $\tilde{\rho}_+$ in the previous Lemma, we only need to estimate $\mathcal{V}$ away for the horizon. This means that for $r \geq \tilde{\rho}_+$ any factors of $\Delta_-$ can now be estimated by large constants $C_{M,l,a}$. In particular, we immediately get

\begin{corollary}
In $r \leq 3M$, we have  (with $g=-f_{n_2}$, $n_2 \geq 2$)
\begin{align} \label{goaly}
- \frac{1}{2}g^\prime \left(\omega-\omega_+\right)^2 + \left(g\tilde{V} \right)^\prime \lesssim  \Delta_-^3 \ k \ \left(L+1\right) e^{-kx} \, .
\end{align}
\end{corollary}
\begin{remark}
The \emph{cubic} degeneration on the right hand side is sufficient for our purposes. The estimate holds for any power of $\Delta_-$. 
\end{remark}
\begin{proof}
Note that $|\tilde{V}| + |\tilde{V}^\prime| \lesssim \mathcal{L} + 1$ in $\tilde{\rho}_+ \leq r \leq 3M$. For $r\leq \tilde{\rho}_+$ the left hand side is negative by the above Lemma.
\end{proof}

Near infinity we have to be careful with the weights in $r$. Nevertheless we have
\begin{lemma}
In $r \geq 3M$, we have (with $g=-f_{n_2}$, $n_2 \geq 2$)
\begin{align} \label{hiu}
- \frac{1}{2}g^\prime \left(\omega-\omega_+\right)^2 + \left(g\tilde{V} \right)^\prime \leq \frac{|2-\alpha|}{l^2} \left[ \frac{n_2-2}{r^{n_2-3}} \frac{1}{l^2}+ \frac{1}{l^2}\frac{k (x_r r^2)}{r^{n_2-2}} \right] e^{-kx}  \nonumber \\ + C_{M,l,a} \frac{k  (x_r r^2) \left(L+1\right)}{r^{n_2-1}} e^{-kx} \nonumber \, \, .
\end{align}
\end{lemma}
\begin{proof}
Near infinity/ in the interior we have (recall $\partial_{r^\star} \sim \frac{r^2}{l^2} \partial_r$ near infinity)
\begin{align}
|\tilde{V} - \frac{2- \alpha}{l^2} \frac{r^2}{l^2}| \lesssim  1+L + |m| |\omega-\omega_+| 
\end{align}
\begin{align}
|\tilde{V}^\prime - \frac{2\left(2-\alpha\right)}{l^2} \frac{r^3}{l^4}| \lesssim r + \frac{1}{r} L + \frac{1}{r} |m| |\omega-\omega_+| 
\end{align}
and hence 
\begin{align}
g \tilde{V}^\prime + g^\prime \tilde{V} \leq -f_{n_2} \left(\frac{2\left(2-\alpha\right)}{l^2} \frac{r^3}{l^4} - \frac{\Delta_-}{r^2+a^2}\left(\frac{n_2}{r} + \frac{k (x_r r^2)}{r^2}\right) \frac{2- \alpha}{l^2} \frac{r^2}{l^2}\right) \nonumber \\
+ C_{M,l,a,\alpha} \left[ g^\prime  \left(1+L+|m||\omega-\omega_+|\right) - g \left(r + \frac{1}{r} L + \frac{1}{r} |m| |\omega-\omega_+| \right) \right] \nonumber \\
 \le \frac{1}{2} g^\prime \left(\omega-\omega_+\right)^2+ \frac{|2-\alpha|}{l^6} \left[ \frac{n_2-2}{r^{n_2-3}} + \frac{k (x_r r^2)}{r^{n_2-2}} \right] e^{-kx}  \nonumber \\
 + C_{M,l,a,\alpha} \left(n_2 + \frac{k  (x_r r^2)}{r} \right)\frac{\left(L+1\right)}{r^{n_2-1}} e^{-kx}    \nonumber
\end{align}
since $m^2 \Xi^2 < L$.
\end{proof}
Combining the previous Lemmata, the estimate arising from the current $Q^{g}_2$ with $g = -f_{n_2}$ finally reads
\begin{align} \label{estimate2}
\int_{R^\star_{-\infty}}^{R^\star_\infty}  dr^\star \Bigg[ \frac{\Delta_-}{\left(r^2 + a^2\right)^2} \left[\left( \frac{n_2}{r^{n_2-1}} + \frac{k \ (r^2 x_r)}{r^{n_2}} \right)  \right] e^{-kx}  \left( |u^\prime|^2 \right) \nonumber \\
+ c_{M,l,a} \cdot k \left(\omega -\omega_+\right)^2 \frac{\Delta_-}{\left(r^2+a^2\right)^2} e^{-kx}  \cdot \frac{1}{r^{n_2}}  |u|^2 \Bigg] \nonumber \\
\leq \int_{R^\star_{-\infty}}^{R^\star_\infty} dr^\star  \Bigg (\frac{|2-\alpha|}{l^6} \left[ \frac{n_2-2}{r^{n_2-3}}  + \frac{k (x_r r^2)}{r^{n_2-2}} \right] e^{-kx} \nonumber \\
+ C_{M,l,a, \alpha}  \left(n_2 + \frac{k  (x_r r^2)}{r} \right)\frac{\left(L+1\right)}{r^{n_2-1}} e^{-kx} \Bigg) \frac{\Delta_-^3}{\left(r^2+a^2\right)^6}  |u|^2 
\nonumber \\ 
+ \left[Q^{g}_2 \left(R^\star_\infty\right) - Q^{g}_2 \left(R^\star_{-\infty}\right)\right] + \Big | \int_{R^\star_{-\infty}}^{R^\star_\infty} dr^\star  \left[ -2f_{n_2} Re \left(u^\prime \overline{H}\right)  \right] \Big| \, .
\end{align}
We now claim that we can absorb (by the left hand side) the $|u|^2$-term on the right hand side in the region $r\geq 5M$ using  Lemma \ref{lem:hardy1}. Indeed, for the first term on the right hand side of (\ref{estimate2}) we apply Lemma \ref{lem:hardy1} with $n=n_2-2$ and $\beta=0$: Choosing $n_2$ sufficiently large (depending only on the size of $|2-\alpha|$) then shows that this term is absorbed using a small part of the good $|u^\prime|^2 \frac{n_2}{r^{n_2-1}}$ term on the left hand side. This also fixes our $n_2$. For the term in the second line of (\ref{estimate2}) (the one which has the constant $C_{M,l,a,\alpha}$), we apply Lemma \ref{lem:hardy1} again, this time with the choice $n=n_2$ and $\beta=1$. The Lemma reveals that for $k$ sufficiently large, the term under consideration can be absorbed by the  $|u^\prime|^2 \frac{k (r^2 x_r)}{r^{n_2}}$ term on the left hand side. We summarize the new estimate as
\begin{align} \label{estimate3}
\int_{R^\star_{-\infty}}^{R^\star_\infty}  dr^\star \Bigg[ \frac{\Delta_-}{\left(r^2 + a^2\right)^2} \left[\left( \frac{n_2}{r^{n_2-1}} + \frac{k \ (r^2 x_r)}{r^{n_2}} \right)  \right] e^{-kx}  \left( |u^\prime|^2 \right) \nonumber \\
+ c_{M,l,a} \cdot k \left( \left(\omega -\omega_+\right)^2 + \cancel{\delta^2} \cdot \mathcal{L}^\star \right)  \frac{\Delta_-}{\left(r^2+a^2\right)^2} e^{-kx}  \cdot \frac{1}{r^{n_2}}  |u|^2 \Bigg] \nonumber \\
\leq C_{M,l,a, \alpha}  \int_{r \leq 5M} dr^\star \ \frac{k  (x_r r^2) \left(L+1\right)}{r^{n_2-1}} e^{-kx}  \frac{\Delta_-^3}{\left(r^2+a^2\right)^6}  |u|^2 
\nonumber \\ 
+ \left[Q^{g}_2 \left(R^\star_\infty\right) - Q^{g}_2 \left(R^\star_{-\infty}\right)\right] + \Big | \int_{R^\star_{-\infty}}^{R^\star_\infty} dr^\star  \left[ -2f_{n_2} Re \left(u^\prime \overline{H}\right)  \right] \Big| \, ,
\end{align}
where we have also used that $\left(\omega -\omega_+\right)^2 \geq 2^{-1} \delta^2 \cdot \mathcal{L}^\star$. (Since $\delta$ has already been fixed in the previous section we can absorb it into $c_{M,l,a}$, which is the reason why it appears cancelled above.)
In the following section, we prove a second multiplier estimate which will allow us to absorb also the part which is in $r\leq 5M$ by the left hand side.

\subsection{Estimating $|u|^2$ from $|u^\prime|^2$}
We apply the current $Q^f_0$ with $f = f_{n_1}$ and $n_1 \geq n_2+1$ sufficiently large so that Lemma \ref{bulklemma} leads to the estimate
\begin{align} \label{fine1}
\int_{R^\star_{-\infty}}^{R^\star_\infty} dr^\star  \left[ \left(\frac{\Delta_-}{(r^2+a^2)^2} \right)^3 \frac{1}{r^{n_1}} \right]  e^{-kx}  |u|^2 \lesssim  \frac{\mathcal{L}^\star}{k^2} \int_{r\leq 3M} dr^\star  \left(\frac{\Delta_-}{r^2+a^2} \right) \frac{e^{-kx}}{r^{n_1-2}}  |u|^2 \nonumber \\
+\frac{1}{k^2} \int_{-R^\star_{-\infty}}^{R^\star_\infty} dr^\star  \left[ \frac{\Delta_-}{\left(r^2+a^2\right)^2}  \frac{1}{r^{n_1-1}} \right] e^{-kx} |u^\prime|^2 
+\frac{1}{k^3} \left[ {Q}^{f_{n_1}}_0 \left(R^\star_\infty\right) - {Q}^{f_{n_1}}_0 \left(R^\star_{-\infty}\right)\right] 
\nonumber \\ 
+ \frac{1}{k^3} \Big | \int_{R^\star_{-\infty}}^{R^\star_\infty} dr^\star  \left[ 2 f_{n_1} Re \left(u^\prime \overline{H}\right) + f_{n_1}^\prime Re \left(u \overline{H}\right) \right] \Big| \, .
\end{align}

\subsection{Putting it together: The basic estimate}

Estimating the integral in $r\leq 5M$ on the right hand side of (\ref{estimate3}) by the estimate (\ref{fine1}) and imposing that $L \ll k^2$ (again with $\ll$ depending only on $M$,$l$ and $a$), we arrive at 
\begin{align} 
\int_{R^\star_{-\infty}}^{R^\star_\infty}  dr^\star  \frac{\Delta_-}{\left(r^2+a^2\right)^2} \frac{e^{-kx}}{r^{n_2}}  \left( |u^\prime|^2 + \left(\omega - \omega_+ \right)^2 |u|^2 + \frac{1}{r} \frac{\Delta_-}{\left(r^2+a^2\right)^2} |u|^2 \right) \lesssim \nonumber \\
 \frac{1}{k^2} \left[ Q^{f_{n_1}}_0 \left(R^\star_\infty\right) - {Q}^{f_{n_1}}_0 \left(R^\star_{-\infty}\right)\right] + \frac{1}{k^2}  \Big | \int_{R^\star_{-\infty}}^{R^\star_\infty} dr^\star  \left[ 2 f_{n_1} Re \left(u^\prime \overline{H}\right) + f_{n_1}^\prime Re \left(u \overline{H}\right) \right] \Big| \nonumber \\ + \frac{1}{k} \left[ {Q}^{-f_{n_2}}_2 \left(R^\star_\infty\right) -{Q}^{-f_{n_2}}_2 \left(R^\star_{-\infty}\right)\right] + \frac{1}{k} \Big | \int_{R^\star_{-\infty}}^{R^\star_\infty} dr^\star  \left[ -2f_{n_2} Re \left(u^\prime \overline{H}\right) \right] \Big| \nonumber
\end{align}
Multiplying by $e^{k\frac{\pi}{2}}$ and using elementary Hardy inequalities near the horizon and near infinity one can immediately improve the weights of the zeroth order term to
\begin{align} 
\int_{R^\star_{-\infty}}^{R^\star_\infty}  dr^\star  \frac{\Delta_-}{\left(r^2+a^2\right)^2} \frac{1}{r^{n_2}}  \left( |u^\prime|^2 + \left(\omega - \omega_+ \right)^2 |u|^2 + r^2|u|^2 \right) 
\lesssim \nonumber \\ \frac{e^{k\frac{\pi}{2}}}{k^2} \left[ Q^{f_{n_1}}_0 \left(R^\star_\infty\right) - {Q}^{f_{n_1}}_0 \left(R^\star_{-\infty}\right)\right]  + \frac{e^{k\frac{\pi}{2}}}{k} \left[ {Q}^g_2 \left(R^\star_\infty\right) - {Q}^{g}_2 \left(R^\star_{-\infty}\right)\right]  \nonumber \\ 
+ \frac{e^{k\frac{\pi}{2}}}{k^2} \Big | \int_{R^\star_{-\infty}}^{R^\star_\infty} dr^\star  \left[ 2 f_{n_1} Re \left(u^\prime \overline{H}\right) + f_{n_1}^\prime Re \left(u \overline{H}\right) \right] \Big| \nonumber \\
 + \frac{e^{k\frac{\pi}{2}}}{k} \Big | \int_{R^\star_{-\infty}}^{R^\star_\infty} dr^\star  \left[ -2f_{n_2} Re \left(u^\prime \overline{H}\right) \right] \Big| \nonumber \, .
\end{align}
Finally, we can add an ``angular" term, $k^{-2} \left(\lambda_{m\ell} + a^2 \omega^2\right) |\Psi|^2$ to the left hand side to obtain our basic integrated decay estimate.
\begin{proposition} 
With $u=\sqrt{r^2+a^2} \psi^{(a\omega)}_{m \ell}$ denoting the Fourier-separated components of $\psi^\tau_{\flat, \omega \not\approx \omega_+}$, we have
\begin{align}  \label{mainintdec}
\int_{R^\star_{-\infty}}^{R^\star_\infty}  dr^\star  \frac{\Delta_- r^{-n_2}}{\left(r^2+a^2\right)^2}   \left[  |u^\prime|^2 +  \left(\left(\omega - \omega_+ \right)^2 + r^2 + r^2 \left[\lambda_{m\ell} + a^2 \omega^2 \right] \right) |u|^2 \right] \nonumber \\ 
\lesssim e^{k\frac{\pi}{2}} \left[ Q^{f_{n_1}}_0 \left(R^\star_\infty\right) - {Q}^{f_{n_1}}_0 \left(R^\star_{-\infty}\right)\right]  + k e^{k\frac{\pi}{2}} \left[ {Q}^g_2 \left(R^\star_\infty\right) - {Q}^{g}_2 \left(R^\star_{-\infty}\right)\right]  \nonumber \\ 
+  e^{k\frac{\pi}{2}} \Big | \int_{R^\star_{-\infty}}^{R^\star_\infty} dr^\star  \left[ 2 f_{n_1} Re \left(u^\prime \overline{H}\right) + f_{n_1}^\prime Re \left(u \overline{H}\right) \right] \Big| \nonumber \\
 + k e^{k\frac{\pi}{2}} \Big | \int_{R^\star_{-\infty}}^{R^\star_\infty} dr^\star  \left[ -2f_{n_2} Re \left(u^\prime \overline{H}\right) \right] \Big|.
 \end{align}
 \end{proposition}
\section{Spacetime estimate for $\psi^\tau_\flat =  \psi_{\flat, \omega \approx \omega_+}^{\tau} + \psi_{\flat, \omega \not\approx \omega_+}^{\tau}$} \label{se:sephif}
We succeeded in proving the microlocal estimates (\ref{mainnonstat}) and (\ref{mainintdec}). We now want to turn these estimates into \emph{spacetime estimates} by integration. For convenience, we will state the estimates in terms of $\psi = u \left(r^2+a^2\right)^{-1/2}$.

\begin{proposition} \label{prop:bases}
The low frequency part of $\psi^\tau$, $\psi^\tau_{\flat}$, satisfies the integrated decay estimate
\begin{align} \label{ibasic}
 \int_{\mathcal{D}} \Bigg[ \left(\partial_t \psi^\tau_{\flat} \right)^2 +  \left( \partial_{r^\star} \psi^\tau_{\flat} \right)^2 +  r^2 |  {}^\gamma \nabla\psi^\tau_{\flat}|_\gamma^2 
   + r^2  \psi^\tau_{\flat} \Bigg]  \frac{1}{r^{n_2+2}} dt^\star dr d\sigma_{\mathbb{S}^2} \nonumber \\
\leq C_{M,l,a} \cdot e^{C_{M,l,a,\alpha} \sqrt{L}} \int_{t^\star=0} J^N_\mu \left[\psi\right] n^\mu
\end{align}
where the integration is over the entire domain of outer communications, $\mathcal{D}$.
\end{proposition}

We will prove this estimate independently for $\psi_{\flat, \omega \approx \omega_+}^{\tau}$ (from (\ref{mainnonstat}))  and $\psi^\tau_{\flat, \omega \neq \omega_+}$ (from (\ref{mainintdec})), hence for the sum. Later, we will improve further both the weights near infinity (using an $r^p$-weighted estimate) and near the horizon (using the redshift) below. 

\begin{proof}
Integrate (\ref{mainnonstat}) and (\ref{mainintdec}) in $\omega$ and sum over $l$ and $m$ to convert the two estimates into spacetime integrated decay estimates for $\psi^\tau_{\flat, \omega \approx \omega_+}$ and $\psi^\tau_{\flat, \omega \not\approx \omega_+}$ respectively. The precise estimation of the terms which arise is contained in Sections \ref{sec:mainterm}-\ref{sec:bndterm}.
\end{proof}
\subsection{The main terms} \label{sec:mainterm}
Replacing $u$ by $\sqrt{r^2+a^2} \psi$, the term on the left hand side of (\ref{mainintdec}) can be estimated
\begin{align}
\lim_{R^\star_{-\infty} \rightarrow -\infty} \int_{-\infty}^\infty d\omega \int_{R^\star_{-\infty}}^{R^\star_\infty}  dr^\star  \sum_{m \ell} \Bigg\{   \frac{\Delta_-}{\left(r^2+a^2\right)^2} \frac{1}{r^{n_2-2}} \nonumber \\
   \Bigg( \Big|\left(\Psi^{(a\omega)}_{m \ell} \right)^\prime\Big|^2 +   \omega^2 |\Psi^{(a\omega)}_{m \ell}|^2 + r^2 \left(\Psi^{(a\omega)}_{m \ell}\right)^2 +  r^2 \left[\lambda_{m\ell} + a^2 \omega^2\right] \left(\Psi^{(a\omega)}_{m \ell}\right)^2 \Bigg)  \Bigg\}\nonumber \\
\gtrsim  \ \int_{\mathcal{D}} \Bigg[  \left(\partial_t \psi_{\flat, \omega \not\approx \omega_+}^\tau \right)^2 +  \left( \partial_{r^\star}\psi_{\flat, \omega \not\approx \omega_+}^\tau  \right)^2 + r^2 |  {}^\gamma \nabla \psi_{\flat, \omega \not\approx \omega_+}^\tau |_\gamma^2 \nonumber \\ + r^2 \left(\psi_{\flat, \omega \not\approx \omega_+}^\tau\right)^2\Bigg]  \frac{\Delta_-}{\left(r^2+a^2\right)^2} \frac{1}{r^{n_2-2}} \sin \theta dt dr^\star d\theta d\tilde{\phi} \, . \nonumber
\end{align}
Here we have used the estimate (a consequence of Lemma \ref{spheroidal})
\begin{align} \label{uhelp}
\int_{-\infty}^\infty \sum_{m\ell} \left(\lambda_{m\ell} \left(a\omega\right) +a^2\omega^2 \right)| \Psi^{(a\omega)}_{m\ell}|^2 d\omega \nonumber \\
\geq  c_{M,l,a} \int_{-\infty}^\infty dt \int_{\mathbb{S}^2} \sin \theta d\theta d\tilde{\phi} |  {}^\gamma \nabla \psi_{\flat}^\tau|_\gamma^2 \, ,
\end{align}
where $\gamma$ denotes the standard unit metric on the $\left(t,r\right)$-spheres. The same estimate holds for the left hand side of (\ref{mainnonstat}) with $\psi_{\flat, \omega \approx \omega_+}^\tau$ replacing $\psi_{\flat, \omega \not\approx \omega_+}^\tau$ everywhere.
\subsection{The $H$-error-terms} \label{sec:Herror}
For the $H$-error-terms in (\ref{mainintdec}) we estimate for a weight function $w\left(r\right)$
\begin{align}
 \lim_{R^\star_{-\infty} \rightarrow -\infty} \sum_{m, \ell} \int_{-\infty}^\infty d\omega \int_{R^\star_{-\infty}}^{R^\star_\infty} dr^\star \ w\left(r\right) Re \left(u^\prime \overline{H} \right) = \nonumber \\
 \lim_{R^\star_{-\infty} \rightarrow -\infty} \sum_{m, \ell} \int_{-\infty}^\infty d\omega \int_{R^\star_{-\infty}}^{R^\star_\infty} dr^\star \ w\left(r\right) Re \left(\Psi^\prime \sqrt{r^2+a^2} \overline{H} + \Psi \frac{r \Delta_-}{\left({r^2+a^2}
 \right)^\frac{3}{2}} \overline{H} \right)  \nonumber \\
\leq  \epsilon \cdot \lim_{R^\star_{-\infty} \rightarrow -\infty} \sum_{m, \ell} \int_{-\infty}^\infty d\omega \int_{R^\star_{-\infty}}^{R^\star_\infty} dr^\star \ \left( r^2 |\Psi|^2 + |\Psi^\prime|^2 \right) \frac{\Delta_-}{r^4} r^{-2p} \nonumber \\
 + \frac{C_{M,l}}{\epsilon} \cdot \lim_{R^\star_{-\infty} \rightarrow -\infty} \sum_{m, \ell} \int_{-\infty}^\infty d\omega \int_{R^\star_{-\infty}}^{R^\star_\infty} dr^\star \  |H|^2 \left[r^{1+p} \cdot w\left(r\right)\right]^2 \frac{r^4}{\Delta_-} \nonumber
 \end{align}
for any $\epsilon>0$ and $p\geq 0$. For the last term, we further estimate
\begin{align}
 \lim_{R^\star_{-\infty} \rightarrow -\infty} \sum_{m, \ell} \int_{-\infty}^\infty d\omega \int_{R^\star_{-\infty}}^{R^\star_\infty} dr^\star \  |H|^2 \left[r^{1+p} \cdot w\left(r\right)\right]^2 \frac{r^4}{\Delta_-} \nonumber \\
 \leq \lim_{R^\star_{-\infty} \rightarrow -\infty} \sum_{m, \ell} \int_{-\infty}^\infty d\omega \int_{R^\star_{-\infty}}^{R^\star_\infty} dr^\star \  |F^{(a\omega)}_{m \ell} |^2 \left[r^{1+p} \cdot w\left(r\right)\right]^2 \Delta_- \cdot  {r^2} \nonumber \\
  \leq \lim_{R^\star_{-\infty} \rightarrow -\infty}\int_{-\infty}^\infty dt \int_{R^\star_{-\infty}}^{R^\star_\infty} dr^\star \int_{\mathbb{S}^2} r^2\sqrt{\gamma} d\theta d\tilde{\phi} \  |F_\flat |^2 \left[r^{1+p} \cdot w\left(r\right)\right]^2 \Delta_- \nonumber \\
  \leq \lim_{R^\star_{-\infty} \rightarrow -\infty} \int_{-\infty}^\infty dt \int_{R^\star_{-\infty}}^{R^\star_\infty} dr^\star \int_{\mathbb{S}^2} r^2\sqrt{\gamma} d\theta d\tilde{\phi} \  r^2 |F |^2 \left[r^{1+p} \cdot w\left(r\right)\right]^2 \frac{\Delta_-}{r^2} \nonumber \\ 
  \lesssim \|\psi\|^2_{H^1_{AdS}\left(\Sigma_0\right)}
\end{align}
provided $|r^{1+p} \cdot w\left(r\right)| \leq C$ is uniformly bounded. Here the last step follows from (\ref{errorloc}). All error-terms appearing in (\ref{mainintdec}), namely
\begin{align}
Re\left(u \overline{H}\right) f_{n_1}^\prime \ \ , \ \ Re \left(u^\prime \overline{H}\right) f_{n_1} \ \ , \ \ -2f_{n_2} Re \left(u^\prime \overline{H}\right) 
\end{align}
can be treated in this way as is readily checked by inspecting the decay of the weights. Clearly, the $H$-error in (\ref{mainnonstat}) is also controlled by $ \|\psi\|^2_{H^1_{AdS}\left(\Sigma_0\right)}$ in view of the previous computation.
\subsection{The boundary-terms} \label{sec:bndterm}
First a rough outline of the argument: The boundary terms at infinity will be easily seen to vanish. The boundary term on $R^\star_{-\infty}$ becomes -- after integration over $\omega, m, \ell$ -- a boundary term on a timelike hypersurface of constant $R^\star_{-\infty}$ on the black hole exterior. In the limit $R^\star_{-\infty} \rightarrow -\infty$ this term becomes a boundary integral on the horizon, which is in turn controlled by the energy of the Hawking-Reall Killing vectorfield.

We now turn to the precise structure of the terms. Note that all boundary terms at $R^\star_{-\infty}$ in both (\ref{mainintdec}) and (\ref{mainnonstat}) are of the form
\begin{align}
{Q}^{...} \left(R^\star_{-\infty} \right) = y \left[ |u^\prime|^2 + \left(\omega - \frac{ma\Xi}{r^2+a^2} \right)^2 |u|^2 \right] + \textrm{terms with $\Delta_-$-factor}
\end{align}
for a bounded function $y=y\left(r^\star\right)$.
Recall that the vectorfield $Z$ defined in \eqref{Zdef} is regular on the future event horizon. In particular, the quantities $Z\left(\psi^\tau\right)$, $Z\left(\psi^\tau_\flat\right)$ are bounded pointwise. In Fourier-space, this condition translates to the condition that 
\begin{align}
\lim_{r^* \rightarrow -\infty} \left[ u^\prime + i\left(\omega - \frac{ma\Xi}{r^2+a^2} \right) u \right] = 0 
\end{align}
holds. Recalling the flux of the Hawking Reall vectorfield through the horizon
\begin{align}
 \int_{\mathcal{H}^+} \left(\partial_{t} \psi + \frac{ma\Xi}{r_+^2+a^2} \partial_{{\phi}} \psi \right)^2 \leq \int_{\mathcal{H}^+} J^K_\mu \left[\psi \right] n^\mu_{\mathcal{H}^+} \,,
\end{align}
one obtains the estimate (``$+0$" indicating terms vanishing in the limit)
\begin{align}
\lim_{R^\star_{-\infty} \rightarrow -\infty} \int_{R^\star_{-\infty}}^{\infty} d\omega \sum_{m \ell} \Bigg\{  - \tilde{Q}^{...} \left(R^\star_{-\infty}\right) \Bigg\}  \nonumber \\
\lesssim \lim_{R^\star_{-\infty} \rightarrow -\infty} \int_{-\infty}^{\infty} dt \int_{\mathbb{S}^2} \sin \theta d\theta d\tilde{\phi} \  |K \psi^\tau_\flat  \left(t,r\left(R^\star_{-\infty}\right),\theta,\tilde{\phi}\right)|^2 + 0
\nonumber \\
\leq \lim_{R^\star_{-\infty} \rightarrow -\infty} \int_{-\infty}^{\infty} dt \int_{\mathbb{S}^2}  \sin \theta d\theta d\tilde{\phi} \  |K \psi^\tau  \left(t,r\left(R^\star_{-\infty}\right),\theta,\tilde{\phi}\right)|^2 + 0
\nonumber \\
= \int_{\mathcal{H}^+\left(0,\tau\right)} | K\left(\psi^\tau\right)|^2  \leq \| \psi \|^2_{H^1_{AdS}\left(\Sigma_0\right)} \, ,
\end{align}
with the last step following from (\ref{hozcon}).
\subsection{Improving the weights near null-infinity}
We can improve the weights near infinity in Proposition \ref{prop:bases}, so that we actually control the full energy integrated in time (of course, it is still a degenerate energy near the horizon). This feature is characteristic of the asymptotically AdS end and was first observed in \cite{gs:stab}.

The following Hardy inequality will be useful, which can be proven using a simple integration by parts and Cauchy-Schwarz (or deduced from Lemma \ref{lem:hardyabsorb}).
\begin{lemma} \label{inftyhardy}
Given $\epsilon>0$, there exists an $\tilde{R}^\star$ depending only on the parameters $a, M, l$ such that
\begin{align}
\int_{\tilde{R}^\star}^{R^\star_\infty} dr^\star r^2 |u|^2 \leq \left[4 - \epsilon\right] l^4 \int_{\tilde{R}^\star}^{R^\star_\infty} dr^\star |u^\prime|^2 \, .
\end{align}
\end{lemma}
\begin{proposition} \label{opti}
We have the estimate

\begin{align} \label{nolk}
 \int_{\mathcal{D}} \Bigg[ \left(\partial_t \psi_{\flat}^\tau \right)^2 +  r^2 \left( \partial_{r^\star} \psi_{\flat}^{\tau} \right)^2 +  r^2 |  {}^\gamma \nabla \psi_{\flat}^{\tau}|_\gamma^2 
   + r^4  \psi_{\flat}^\tau  \Bigg]  \frac{1}{r^{2}} dt^\star dr d\sigma_{\mathbb{S}^2} 
   \nonumber \\
\leq C_{M,l,a} \cdot e^{C_{M,l,a} \sqrt{L}} \int_{t^\star=0} J^N_\mu \left[\psi\right] n^\mu \, .
\end{align}
\end{proposition}
Remark: Inspecting the proof of Proposition \ref{prop:bases}, we see that we already proved this for $\psi^\tau_{\flat,\omega \approx \omega_+}$.  
\begin{proof}
Choose $f = -\frac{1}{r}$ for the microlocal current ${Q}_0^f$ in (\ref{fcurrent}). Repeating the conversion to physical space of the previous section, one sees that the boundary terms near infinity and the horizon can be dealt with as before using only the boundedness statement.  What does the resulting spacetime term control? Inspecting (\ref{fbulk}), we see that
\begin{align}
\left(Q^f_0\right)^\prime \geq \frac{2\Delta_-}{r^2 \left(r^2+a^2\right)} |u^\prime|^2 - C_{M,l,a} \left(1+\frac{L}{r}\right)  |u|^2 + 2\left(2- \alpha\right) \frac{r^2}{l^6} |u|^2  - \textrm{H-Err} \nonumber \, .
\end{align}
Let us choose an $r_L>r_{+}$ so large that both
\begin{align}
 \frac{\Delta_-}{r^2 \left(r^2+a^2\right)} \geq \frac{1}{l^2} + \frac{1}{2r^2}
\end{align}
and also
\begin{align}
\frac{1}{2} \left[\frac{1}{2} + 2\left(2-\alpha\right) \right] \frac{r^2}{l^6} \geq C_{M,l,a} \left(1+\frac{L}{r}\right) 
\end{align}
hold for $r \geq r_L$. Note that (besides the dependence of $M$, $l$ and $a$) $r_L$ will depend polynomially (in fact, linearly) on $L$, and like $\left(9/4-\alpha\right)^{-1}$ on the mass. 
We now apply Lemma \ref{inftyhardy} (possibly making $r_L$ a bit larger), which guarantees
\begin{align}
\int^{R^\star_\infty}_{r_L^\star} \left({Q}^f_0\right)^\prime dr^\star \geq  c_{M,l,a,\alpha} \int^{R^\star_\infty}_{r_L^\star} 2 \frac{\Delta_-}{r^2} \left( |\Psi^\prime|^2 + r^2 \Psi^2\right) dr^\star   - \textrm{H-Err} \, . \nonumber
\end{align}
In the region $r\leq r_L$ we can estimate the $r$-weights by $L$-weights to find
\begin{align}
\Big| \int_{R^\star_{-\infty}}^{r_L^\star} \left({Q}^f_0\right)^\prime dr^\star \Big| \lesssim  L^{n_2}  \int_{R^\star_{-\infty}}^{r_L^\star}  \frac{2\Delta_-}{\left(r^2+a^2\right)^2}\frac{1}{r^{n_2-2}} \left( |\Psi^\prime|^2 + r^2 \Psi^2\right) dr^\star   - \textrm{H-Err} \nonumber  ,
\end{align}
with the first term being already controlled from our basic estimate. The fact that we lose polynomially in $L$ here is irrelevant in view of the exponential loss which has already occurred.

Finally, the $H$-error terms can be treated as previously, namely applying the estimates of Section \ref{sec:Herror} with $w\left(r\right) = r^{-1}$ and $p=0$). Hence the estimate of the Proposition follows.
\end{proof}

\subsection{Improving the weights near the horizon} \label{sec:hozimprove}
We would like to improve the weights near the horizon in Proposition \ref{opti} so that the left hand side is actually the non-degenerate energy integrated in time. However, we cannot construct
a microlocal redshift estimate applying the techniques by which we derived the integrated decay estimate of Proposition \ref{opti}. The reason is that the resulting boundary term will not have a good sign near the horizon (since it will be supported on both the future and the past horizon). One could of course estimate this boundary term by the (non-degenerate) $N$-energy of $\psi^\tau$ on the future event horizon. However, the latter actually diverges in the limit as $\tau \rightarrow \infty$, since the full solution $\psi$ only decays logarithmically. This is in sharp contrast with the situation, when we derived the integrated decay estimate of Proposition \ref{opti}: Here the boundary term could be localized on the future event horizon and controlled by the (degenerate!) $K$-energy of $\psi^\tau$, which \emph{is} controlled from the boundedness statement.\footnote{This difficulty is absent in the asymptotically flat case, where one proves a local integrated decay estimate 
 for the full solution. The problem here arises precisely from the fact that we cannot prove an integrated decay estimate for the high-modes because of the ``stable trapping"!}

The resolution is to apply the $N$-estimate in physical space locally between two spacelike $\Sigma_{t^\star}$-slices. The only problem, then, is to get a good bound for the $N$-energy of $\psi^\tau_\flat$ (from the total energy) on some initial slice. This will be achieved later using the pigeonhole principle. We start with the local $N$-estimate for $\psi^\tau_\flat$:

\begin{proposition} \label{localN}

\begin{align} \label{dir}
  \int_{\Sigma_{\tau_2}} J^N_{\mu}\left[\psi_\flat^\tau\right] n^\mu + \int_{\mathcal{R}\left(\tau_1,\tau_2\right)} J^N_{\mu}\left[\psi_\flat^\tau\right] n^\mu \lesssim   \int_{\Sigma_{\tau_1}} J^N_{\mu}\left[\psi_\flat^\tau\right] n^\mu \nonumber \\
  + C_{M,l,a} \cdot e^{C_{M,l,a} \sqrt{L} \frac{\pi}{2}} \int_{t^\star=0} J^N_\mu \left[\psi \right] n^\mu
\end{align}
holds for any $\tau_2 \geq \tau_1 \geq 0$.
\end{proposition}
\begin{proof}
Apply the $N$-estimate for $\Box_g \psi^\tau_\flat + \frac{\alpha}{l^2} \psi^\tau_\flat = F^\tau_\flat$ in physical space from $\tau_1$ to $\tau_2$. The spacetime term $K^N\left[\psi_\flat^\tau\right]$ has a good sign near the event horizon (except for the zeroth-order mass term) while away from it it is estimated from Proposition \ref{opti} (which also takes care of the zeroth-order mass-term). Add the estimate (\ref{reold}) and  the estimate of Proposition \ref{opti} to the resulting estimate. This yields (\ref{dir}) except for the error term
\begin{align}
\int_{\mathcal{R}\left(\tau_1,\tau_2\right)} F_\flat^\tau N\left(\psi_\flat^\tau\right)
\end{align}
on the right hand side. The latter is controlled by borrowing an $\epsilon$ from the spacetime term on the left and controlling $\int_{\mathcal{R}\left(\tau_1,\tau_2\right)} |F_\flat^\tau|^2$ from the boundedness statement as in Section \ref{sec:Herror}.
\end{proof}

\begin{remark}
Using elementary spectral analysis (involving Weyl's law and compactness arguments) one can show that the first term on the right hand side of (\ref{dir}) is always finite (with a constant that may blow up as $L \rightarrow \infty$). This is not needed in the following, as we are going to apply the above Proposition only for special $\Sigma_{\tau_1}$, for which a better, completely explicit, bound for the right hand side is available.
\end{remark}

\subsection{Higher derivatives} \label{higherder}
The estimates of Propositions \ref{opti} and \ref{localN} can be obtained for all higher derivatives in by commuting globally with $T$ and locally near the horizon with the redshift vectorfield (as in \cite{Mihalisnotes, HolzegelAdS}) and locally near infinity with angular momentum operators. Since these estimates have become standard after \cite{Mihalisnotes} and have appeared in detail previously in the literature, they are omitted.
\section{Estimates for $\psi^\tau_{\sharp}$} \label{se:esph}
Recall the densities (\ref{endensdef}) and the energies defined at the end of Section \ref{sec:norms}. 
\begin{proposition} \label{propsharp}
We have the estimate
\begin{align}
\int_{\mathcal{D}} e_1\left[\psi_\sharp^\tau\right] \leq \frac{1}{L} E_2\left[\psi\right] \cdot \tau
\end{align}
\end{proposition}

\begin{proof}
From (\ref{e2est}) we already have
\begin{align} 
\int_{\mathcal{D}}  \sqrt{g} dt^\star dr d\theta d\phi \  e_2\left[\psi^\tau_\sharp\right] 
\leq E_2\left[\psi\right] \cdot \tau
\end{align}
Now we can use that $L \leq \lambda_{m\ell} \left(a\omega\right) +a^2\omega^2$ holds for $\psi^\tau_{\sharp}$. For instance, for $\partial_{\tilde{\phi}}$,
\begin{align} \label{uhelp2}
L \int_{-\infty}^\infty dt \int d\sigma_{\mathbb{S}^2} |  \partial_{\tilde{\phi}} \psi_{\sharp}^\tau|_\gamma^2 = L \int_{-\infty}^\infty \sum_{m\ell} | im \widehat{\psi_\sharp^\tau}^{(a\omega)}_{m\ell}|^2 d\omega \leq \nonumber \\
\int_{-\infty}^\infty \sum_{m\ell} \left(\lambda_{m\ell} \left(a\omega\right) +a^2\omega^2 \right)| \widehat{\left(\partial_{\tilde{\phi}} \psi_\sharp^\tau\right)}^{(a\omega)}_{m\ell}|^2 d\omega \nonumber \\
\leq  C_{M,l,a} \int_{-\infty}^\infty dt \int d\sigma_{\mathbb{S}^2} \left[ |  {}^\gamma \nabla \partial_{\tilde{\phi}} \psi_{\sharp}^\tau|_\gamma^2 + | \partial_t \partial_{\tilde{\phi}} \psi_{\sharp}^\tau|^2 \right] \, ,
\end{align}
and similarly with the derivatives $\partial_r$, $\partial_\theta$ and $\partial_t$ (and linear combinations like the regular $Z$-derivative  defined in (\ref{Zdef})).
Multiplying (\ref{uhelp2}) (and its analogues for the other derivatives) with their appropriate $r$-weight and integrating over the domain of outer communications, we obtain
\begin{align} \label{eq:1go}
\int_{\mathcal{D}} e_1\left[\psi_\sharp^\tau\right] \leq \frac{1}{L}\  \int_{\mathcal{D}}  \sqrt{g} dt^\star dr d\theta d\phi \  e_2\left[\psi^\tau_\sharp\right] 
\end{align}
from which the claim follows.
\end{proof}

\section{Proof of Theorem \ref{theorem1}: Logarithmic decay} \label{se:pth1}
Starting from the estimate
\begin{align}
\int_{\mathcal{R}\left(0,\tau/2\right)} r^2 \sin \theta \ dt^\star dr d\theta d\phi \ e_1 \left[\psi_\flat^\tau\right] \leq  \int_{\mathcal{D}} e_1\left[\psi_\flat^\tau \right] \leq  \int_{\mathcal{D}} e_1\left[\psi^\tau\right] \leq E_1\left[\psi\right] \cdot \tau \, \nonumber ,
\end{align}
we find a slice $\Sigma_{\tau^\prime}$ with $0< \tau^\prime< \frac{\tau}{2}$, such that
\begin{align}
\int_{\Sigma_{\tau^\prime}} r^2 \sin \theta \ dr d\theta d\phi \  e_1 \left[\psi_\flat^\tau\right]  \lesssim E_1\left[\psi\right] \, .
\end{align}
We apply Proposition \ref{localN} from that slice $\Sigma_{\tau^\prime}$ to any slice to the future of it to obtain in particular
\begin{align} 
\int_{\mathcal{R}\left(\tau^\prime,\tau-1\right)} r^2 \sin \theta \ dt^\star dr d\theta d\phi \ e_1\left[\psi_\flat^\tau \right] \lesssim  
  C_{M,l,a} \cdot e^{C_{M,l,a} \sqrt{L} \frac{\pi}{2}} \int_{t^\star=0} J^N_\mu \left[\psi \right] n^\mu \nonumber \, .
\end{align}
Adding the estimate of Proposition \ref{propsharp}, we find
\begin{align}
\int_{\mathcal{R}\left(\tau/2,\tau-1\right)} e_1\left[\psi^\tau\right] \leq \int_{\mathcal{R}\left(\tau/2,\tau-1\right)} \left(e_1\left[\psi^\tau_\flat \right] + e_1\left[\psi^\tau_\sharp\right] \right)  \lesssim   \left( e^{C_{M,l,a}  \sqrt{L}}  + \frac{\tau}{L} \right) E_2\left[\psi\right] \nonumber \, .
\end{align}
But the strip we are integrating over has length $\sim \frac{\tau}{2}$ and we can therefore extract a good slice between $\frac{\tau}{2} \leq \tau_g \leq \tau-1$ where
\begin{align} \label{es:interpolate}
\int_{\Sigma_{\tau_g}} r^2 \sin \theta \ dr d\theta d\phi  \  e_1 \left[\psi^\tau\right]  \lesssim \left( \frac{e^{C_{M,l,a} \sqrt{L}}}{\tau}  + \frac{1}{L} \right) E_2\left[\psi\right] \, .
\end{align}

This holds for any $L$, in particular for $\sqrt{L} = \frac{1}{2 C_{M,l,a}} \left( \log \tau \right)$, which implies that there is a slice $\Sigma_\tau$ between $\frac{\tau}{2}$ and $\tau-1$ such that the energy decays like $\left(\log \tau\right)^{-2}$. Hence we have obtained a dyadic sequence of good slices through which the energy decays like $\left(\log \tau\right)^{-2}$.

\begin{proposition}
The non-degenerate energy
\begin{align}
f\left(\tau\right) = \int_{\Sigma_{\tau}} e_1 \left[\psi\right] r^2 \sin \theta dr d\theta d\phi \sim  \int_{\Sigma_{\tau}} J^N_\mu \left[\psi\right] n^\mu  
\end{align}
satisfies 
\begin{align}
f\left(\tau\right) \lesssim \frac{E_2\left[\psi\right]}{\left(\log \tau\right)^2} 
\end{align}
for all $\tau \geq 2$.
\end{proposition}

\begin{proof}
Note that the statement is immediate for the degenerate Hawking-Reall energy $\int_{\Sigma_{\tau}} J^K_{\mu} n^\mu_{\Sigma}$ in view of (\ref{conslaw}): The $\left(\log \tau\right)^{-2}$-decay is exported dyadically (via energy conservation from the dyadic sequence of good slices already obtained) implying that the \emph{degenerate} energy decays like $\left(\log \tau\right)^{-2}$ through any slice.

This allows us to revisit the estimate for the vectorfield $N$, cf.~(\ref{reold}). Namely, we estimate the error-term in the interior (away from the horizon) using the $K$-energy which has already been shown to decay logarithmically to obtain the local estimate ($\tau_2 \geq \tau_1 \geq 2$):
\begin{align} \label{uyt}
f\left(\tau_2\right) + \int_{\tau_1}^{\tau_2} d\tau f\left(\tau\right) \leq C \cdot f\left(\tau_1\right) + C  \max \left(\left(\tau_2-\tau_1\right),1\right) \frac{E_2\left[\psi\right]}{\left( \log \tau_1\right)^2} \, 
\end{align}
for a $C>1$ depending only on the black hole parameters and $\alpha$. We claim that (\ref{uyt}) 
implies that 
\begin{align}
f\left(\tau\right) \leq \tilde{C} \cdot E_2\left[\psi\right] \left(\log \tau\right)^{-2}
\end{align}
for $\tau \geq 2$ using a bootstrap argument. Here is a sketch: Define $\mathfrak{l}= 8C^2$. Consider the region where $f\left(\tau\right) \leq A\cdot E_2\left[\psi\right]\left(\log \tau\right)^{-2}$ for $A=32 C^3$. Decompose this region into strips of length $\mathfrak{l}$ (the last one may be shorter). In each strip, i.e.~between $\tau_n$ and $\tau_{n+1}$, we find a good slice for which
\begin{align}
f\left(\tau_{good}\right) \leq \frac{1}{\mathfrak{l}} C \frac{A\cdot E_2\left[\psi\right]}{\left( \log \tau_n\right)^2} + C \frac{E_2\left[\psi\right]}{\left(\log \tau_n\right)^2} = \left(\frac{CA}{\mathfrak{l}} + C\right) \frac{E_2\left[\psi\right]}{\left(\log \tau_n\right)^2} \, .
\end{align}
Applying the estimate from the good slice up to any slice to the past of the next good slice we find
\begin{align}
f\left(\tau\right) \leq \left(A \frac{C^2}{\mathfrak{l}} + C^2 + C \ 2\mathfrak{l} \right) \frac{E_2\left[\psi\right]}{\left(\log \tau_n\right)^2} < \frac{3A\cdot E_2\left[\psi\right]}{4\left(\log \tau_n\right)^2} < \frac{A\cdot E_2\left[\psi\right]}{\left(\log \tau\right)^2} \, ,
\end{align}
where the last step uses that we can assume that $\tau_n \gg \mathfrak{l}$ and hence $\frac{1}{\tau_n} \sim \frac{1}{\tau_{n+2}}$ with a constant very close to $1$. This improves the bootstrap assumption.
\end{proof}

\begin{remark}
By repeated commutation one can further improve (\ref{eq:1go}) to hold with $L^{-n}$ replacing $L^{-1}$ and $e_{n+1}\left[\psi^\tau_\sharp\right]$ replacing $e_{2}\left[\psi^\tau_\sharp\right]$ for $n>1$, with the appropriate definition for $e_{n+1}\left[\psi^\tau_\sharp\right]$. (For instance, it suffices to define recursively $e_{n+1}\left[\psi\right]=e_n\left[\partial_t \psi\right] + \sum_i e_n\left[\Omega_i \psi\right]+e_n\left[\psi\right]$ for $n>1$, and to combine this with elliptic estimates from the energies produces by these densities.) The interpolation in (\ref{es:interpolate}) then yields $\left(\log \tau\right)^{-2n}$-decay. We have hence the following corollary:
\end{remark}

\begin{corollary}
In addition to (\ref{es:maines}), the more general estimate
\begin{align} \label{es:maines2}
\| \psi \|_{H^1_{AdS}\left(\Sigma_{t^\star}\right)} \lesssim C_n \left(\log \ t^\star \right)^{-n} \left( \int_{\Sigma_0} e_{n+1}\left[\psi\right] r^2 \sin \theta dr d\theta d\phi \right)^\frac{1}{2}
\end{align}
also holds for $n\geq 1$.
\end{corollary}

%Reflecting back on the proof, we see that the key element was the construction of the exponential multipliers in Section \ref{phiflatt}. This construction depended only on the behavior of the metric at the ``ends" (the horizon and infinity) and the fact that the boundary-terms could be controlled from a boundedness statement.
%
%
%
\section{Some remarks on the Schwarzschild-AdS case} \label{se:srsc}
While we are not able to improve the logarithmic decay rate in the case of Schwarzschild-AdS (in accord with the heuristics given in the introduction), one can nevertheless extract some more information about the solution in this special case. The underlying reason is that for $a=0$ the background is spherically symmetric and one can expand the solution into the familiar spherical harmonics, $Y_{lm}$, which have eigenvalues $\ell\left(\ell+1\right)$ (with a $2\ell+1$ degeneration for fixed $\ell$) \emph{before} cutting off in time. One separates the solution into a part for which $\ell \left(\ell+1\right) \leq L$ and a part for which $\ell \left(\ell+1\right) \geq L$:
\begin{align}
\psi_\flat \left(t,r,\theta,\tilde{\phi}\right) &= \sum_{m \ell} \int_{\mathbb{S}^2} \sin \underline{\theta} d\underline{\theta} d\underline{\tilde{\phi}} \ \psi \left(t,r,\underline{\theta}, \underline{\tilde{\phi}} \right)  Y_{\ell m} \left(\underline{\theta}, \underline{\tilde{\phi}}\right) Y_{\ell m} \left(\theta, \tilde{\phi}\right) \chi \left(\frac{\ell \left(\ell+1\right)}{L}\right) \nonumber \\
\psi_\sharp  \left(t,r,\theta,\tilde{\phi}\right)  &= \psi  \left(t,r,\theta,\tilde{\phi}\right)  - \psi_{\flat} \left(t,r,\theta,\tilde{\phi}\right)  \nonumber
\end{align}
where $\chi$ is the familiar cut-off functions defined in Section \ref{timecutoff}. Because there is no cut-off in time, both $\psi_{\flat}$ and $\psi_{\sharp}$ satisfy the \emph{homogeneous} wave equation. 
Moreover, since the non-degenerate energies of $\psi_\flat$ are easily seen to be controlled (uniformly in $L$) by the full $\psi$-energies, one obtains separate uniform boundedness statements for $\psi_\flat$ and $\psi_\sharp$ for free.\footnote{Note that this was not possible in the Kerr-AdS-case where the coupling of the $S_{m\ell}$ with $\omega$ prevented one from estimating the $\psi_\flat$-energy on a spacelike slice from the $\psi$-energy. It is an interesting question whether a more refined spectral analysis of the $S_{m\ell}$ can provide such an estimate.} 

With this in mind, starting from $\Box_g \psi_{\flat} + \frac{\alpha}{l^2} \psi_{\flat} = 0$, one cuts $\psi_{\flat}$ off in time and repeats the previous analysis for $\psi_{\flat}$. Since the time-cut-off commutes with the cut-off in angular momentum in the spherically-symmetric case, we now obtain the following version of Proposition \ref{localN}:
\begin{align} 
  \int_{\Sigma_{\tau_2}} J^N_{\mu}\left[\psi_\flat\right] n^\mu + \int_{\mathcal{R}\left(\tau_1,\tau_2\right)} J^N_{\mu}\left[\psi_\flat\right] n^\mu \lesssim   \int_{\Sigma_{\tau_1}} J^N_{\mu}\left[\psi_\flat\right] n^\mu \nonumber \\
  + C_{M,l,a} \cdot e^{C_{M,l,a} \sqrt{L} \frac{\pi}{2}} \int_{\Sigma_{\tau_1}} J^N_\mu \left[\psi_\flat \right] n^\mu \, .
\end{align}
Note that the last term on the right now involves only $\psi_\flat$ and integration over $\Sigma_{\tau_1}$. The above estimate implies that $\psi_\flat$ decays exponentially:
\begin{align}
 \int_{\Sigma_{t^\star}} J^N_{\mu}\left[\psi_\flat\right] n^\mu \lesssim \exp \left(- e^{-C_{M,l,a} \sqrt{L}} \cdot t^\star \right) \, .
\end{align}
Despite this strong decay, interpolating with $\psi_\sharp$ (for which only the boundedness statement is available) one cannot improve the logarithmic decay rate for the full solution. The question whether the low modes in the Kerr-AdS case also decay exponentially remains open.

% We finally remark that the exponential type multiplier estimates proven are sufficiently robust to establish this type of decay for any stationary spherically-symmetric black hole solution (asymptotically-flat or asymptotically AdS). This is because the exponential estimates only depend on the behavior at the ends, i.e.~the asymptotics near the horizon and near infinity. The precise set-up and statements will be elaborated on in future work.

%

\section{Proof of Theorem \ref{theorem2}: Absence of periodic solutions} \label{noperiodic}

We give an elementary proof that there exist no time-periodic solutions for the massive wave equation on Kerr-AdS spacetimes in the parameter range $|a|l < r_+^2$. This result is an immediate 
consequence of the $\left(\log t^\star\right)^{-1}$-decay established in the main theorem. However, whereas the main theorem required the construction of elaborate multipliers, the absence of periodic solutions can be obtained from elementary ODE arguments and the redshift effect near the horizon only. It is for this reason together with the fact that the latter result holds for a wider range of parameters that we include its proof here.
\begin{proof}[Proof of Theorem \ref{theorem2}.]
Step 1: We first prove that, as in the Schwarzschild case (cf.~the Appendix of \cite{gs:stab}), there exists no \emph{stationary} (with respect to the globally timelike Hawking Reall vectorfield $K$) solution. This is accomplished by the usual ``Bekenstein"-argument \cite{Bekenstein, Heusler} combined with a Hardy inequality for the mass term. Assuming that $\psi$ is a solution of $\Box \psi + \frac{\alpha}{l^2}\psi=0$ which satisfies $K\left(\psi\right)=T\left(\psi\right) + \lambda \Phi \left(\psi\right) =0$ for $\lambda = \frac{a\Xi}{r_+^2 + a^2}$, the wave equation reduces in Boyer-Lindquist coordinates to:
\begin{eqnarray} \label{redux}
\frac{\Xi}{\Sigma |\sin \theta |}\Big[ \partial_r \left( \sqrt{-g} g^{rr} \partial_r \psi \right)+\partial_\theta \left( \sqrt{-g}g^{\theta \theta} \partial_\theta \psi \right)\Big]
\nonumber \\
+\left(g^{\tilde{\phi} \tilde{\phi}} - 2\lambda g^{t \tilde{\phi}} + \lambda^2 g^{tt} \right) \partial^2_{\tilde{\phi}} \psi  +\frac{\alpha}{l^2}\psi=0 \, .
\end{eqnarray}
The expression in the round bracket of the second line is 
\begin{align}
P = g^{\tilde{\phi} \tilde{\phi}} - 2\lambda g^{t \tilde{\phi}} + \lambda^2 g^{tt} &= \frac{1}{det A} \left(g_{tt} + 2 \lambda g_{t \tilde{\phi}} + \lambda^2 g_{\tilde{\phi} \tilde{\phi}} \right) \nonumber 
= \frac{g\left(K,K\right)}{det A}  \ge 0 \, , \nonumber \\
\textrm{where \ \ } det A &= g_{tt} g_{\tilde{\phi} \tilde{\phi}} - \left(g_{t \tilde{\phi}}\right)^2=-\frac{\Delta_- \Delta_{\theta}}{\Xi^2}\leq 0 \, , \nonumber
\end{align}
since the Hawking-Reall vectorfield is timelike in the black hole exterior.
Note that the vector field $\partial_r$ (associated to Boyer-Lindquist coordinates) is given, in terms of $(\partial_{t^\star},\partial_\phi, R)$, where $R(f)=\frac{\partial f}{\partial r}|_{(t^*,\theta,\phi)=const}$ is the radial vector field of the $(t^\star,r,\theta,\phi)$ coordinates, by:
\begin{align}
\partial_r &= \frac{1}{\Delta_{-}}\left( \frac{2Mr}{1+\frac{r^2}{l^2}}\partial_{t^\star}+ a \Xi \partial_{\phi} \right)+R 
         \nonumber \\
          &= \frac{1}{\Delta_{-}}\frac{2Mr}{1+\frac{r^2}{l^2}} K  - \frac{2Ma\Xi \left(r-r_+\right)}{\Delta_- \left(\left(r^2+a^2\right) \left(1+ \frac{r^2}{l^2}\right)\right)} \partial_\phi + R, \nonumber
\end{align}
where $K$ is the Hawking-Reall vector field. Since $K(\psi)=0$ and since the last two terms are regular vectorfields on the event horizon, it follows that $\partial_r \psi$ extends continously to any sphere arising from the intersection of a $t^*=const$ slice and the future event horizon. %In particular, the $\Delta_- \partial_r \psi$ vanishes on 

Multiplying (\ref{redux}) by $\frac{\Sigma |\sin \theta |}{\Xi} \psi$, we find after integration by parts on a $t^*=const$ slice\footnote{The integration is performed on $t^*=const$ slice to avoid the bifurcate sphere $\{t=const\} \cap \{r=r_+\}$ where all the coordinate systems that we have so far introduced break down.}:
\begin{eqnarray} \label{eq:posneg}
\int_{r=r_+}^\infty\int_{\theta,\phi}dr d\theta d\phi \Big\{ \sqrt{-g}\left( g^{rr} (\partial_r \psi)^2+g^{\theta \theta}(\partial_\theta \psi)^2+P (\partial_{\tilde{\phi}} \psi)^2\right) \nonumber \\ 
-\frac{\alpha}{l^2}\psi^2  \frac{\Sigma |\sin \theta|}{\Xi}\Big\}(t^*,r,\theta,\phi)=0,
\end{eqnarray}
where the boundary term at infinity has vanished in view of the decay of $\psi$ and the boundary term at the horizon vanished since $g^{rr}$ extends continuously to $0$ at the event horizon and since $\psi \partial_r \psi$ is regular at $r_+$, as established above. %because $\psi$ is bounded and $\lim_{r\rightarrow r_+} \Delta_- \partial_r \psi=0$, as demonstrated above.

As in \cite{HolzegelAdS, gs:stab}, using the bound $|a| < \frac{l}{2}$ (or any of the conditions given in Theorem \ref{theorem1}) , the zero order term can be absorbed into the first term via a Hardy inequality, from which it follows that $\psi=0$. 

\begin{remark}
The recent results of \cite{GHHCMW} establish that the expression on the right hand side of (\ref{eq:posneg}) is zero if and only if $\psi=0$ using only the conditions $|a|<l$ and $r_+^2 > |a|l$. As the remainder of the proof also only requires these two conditions, Theorem \ref{theorem2} actually holds for the full range of admissible parameters.
\end{remark}

Step 2:  We use a microlocal decomposition to reduce the analysis to that of an ode with trivial data. More precisely, in view of the separability of the geodesic equation, we may restrict our attention to solutions $\psi$ of the form:
\begin{align} \label{ansatz}
\psi=e^{-i \omega t} f(r^*) S_{m\ell}( a\omega,\cos \theta)e^{i m \tilde{\phi}},
\end{align}
where $S_{m\ell}( a\omega,\cos \theta)e^{i m \tilde{\phi}}$ are the (modified) oblate spheroidal harmonics introduced in Section \ref{se:sphhar}.
In view of Step 1 we can assume $\omega-\omega_+\neq 0$.
Since $\psi$ is regular at the future event horizon and since the vector field $-\Delta_- Z=\partial_t+\frac{a \Xi}{r_+^2+a^2}\partial_{\phi}-\partial_{r^*}$ vanishes on $\mathcal{H}^+$ (cf.~(\ref{Zdef})), it follows that:
$$
f'+i\left(\omega - \frac{a \Xi m}{r_+^2+a^2}\right)f=0.
$$
In fact, we have on any $t^*=const$ slice, $-\Delta_- Z= \Delta_- D$ where $D$ is a regular vector field at the horizon. 
Hence we obtain that $$f'+i\left(\omega - \frac{a \Xi m}{r_+^2+a^2}\right)f= \mathcal{O}\left(\Delta_- \right),$$ near the horizon.
Since $\frac{dr}{dr^*}=\frac{\Delta_-}{r^2+a^2}$ vanishes at the horizon, the rescaled function $u=\sqrt{r^2+a^2} f$ also satisfies
\begin{align} \label{oov}
u'+i\left(\omega - \frac{a \Xi m}{r_+^2+a^2}\right)u= \mathcal{O}\left(\Delta_- \right).
\end{align}
Moreover, since $\psi$ is in the class $CH^1_{AdS}$, we have the following uniform\footnote{A priori, for a general solution, the $r$ decay is only locally uniform in the $t$ variable. Here, of course, we have global uniform radial decay in view of the ansatz.} radial decay estimates (recall that $'$ denotes $r^*$-differentiation):
\begin{eqnarray}
\psi r^{3/2}=o(1) \textrm{ \ \ \ \ , \ \ \ \ }
\psi' r^{1/2}=o(1).
\end{eqnarray}
Hence, we find for $u=\sqrt{r^2+a^2} \psi$, the decay rates
\begin{eqnarray} \label{udecrates}
u r^{1/2}=o(1) \textrm{ \ \ \ \ , \ \ \ \ }
u' r^{-1/2}=o(1).
\end{eqnarray}
In Section \ref{se:sepv}, we obtained the following ode \eqref{eq:uode} for each $u^{(a\omega)}_{m\ell}$, which we will write short hand as $u$:
\begin{align} \label{sode}
u''+ \left(\omega^2-V\left(r\right) \right)u=H.
\end{align}
Compared to \eqref{eq:uode}, here we have $H=0$ since no cut-off is needed in view of the ansatz \eqref{ansatz}. Recall the current (\ref{Tcurr}).
From the decay of $u$ at $\infty$, it follows that 
\begin{align} \label{imaux}
Q_T[u]=0   \ \ \  \ \ \ \textrm{and hence}  \ \ \ Im \left(u^\prime \bar{u}\right) = 0
\end{align}
globally, which in conjunction with the boundary condition (\ref{oov}) on the horizon implies that $u$ vanishes on the horizon.

Next we choose the free function $z$ in the red-shift current (\ref{redshift}) to be $z=V_{red}^{-1}$, where $V_{red}$ was defined in (\ref{vred}). % Note that $\tilde{V}=\mathcal{O}(r-r_+)$.
Since 
$$\left| u' +i\left(\omega - \omega_+ \right)u \right|^2=\mathcal{O}(r-r_+)^2$$ 
from (\ref{oov}) and $u=0$ on the horizon, it follows that $Q_{red}^{V^{-1}_{red}}=0$ on the horizon. The derivative satisfies (use (\ref{imaux}))
\begin{align} \label{Qreddec} 
Q_{red}'= -\frac{V_{red}'}{V_{red}^2}\left| u' +i\left(\omega - \omega_+\right)u \right|^2 -\frac{1}{V_{red}}(|u|^2)'\left(V_{red}-\tilde{V}\right)\leq 0
\end{align}
in a neighborhood of the horizon, which implies that $Q_{red}^{V_{red}^{-1}}$ is non-increasing near the horizon. On the other hand, one easily checks that $Q_{red}^{V_{red}^{-1}} \ge 0$ holds near the horizon, which is only consistent with (\ref{Qreddec}) if $Q_{red}^{V_{red}^{-1}}$ and hence $u$ are identically zero in a neighbourhood of the horizon. Solving the linear ode (\ref{sode}) with trivial initial conditions we conclude $u=0$ identically.
\end{proof}

\section{Acknowledgements}
The authors thank Mihalis Dafermos and Igor Rodnianski for insightful discussions. We also acknowledge the hospitality of the Erwin Schr\"odinger Institute, Vienna, where part of this work was being completed in July 2011 during the workshop ``Dynamics of General Relativity". J.S.~thanks the Albert Einstein Institute, Potsdam for financial support during his postdoctoral fellowship. G.H. thanks NSF for support in DMS-1161607.

%%      ---------------------------------------------------------------------
%%      ------------------------- APPENDIX (OPTIONAL) -----------------------
%%      ---------------------------------------------------------------------
        
%%      If you have one appendix, uncomment the line \appendix and add
%%      a \section{ *** APPENDIX TITLE ***}. If you have more than
%%      one, uncomment the line \appendices and add a \section{ ***
%%      APPENDIX TITLE ***} command for each appendix title.

\appendix
%\appendices
\section{Basic Fourier estimates} \label{ap:suf}
Let $\varphi \left(t,r,\theta,\tilde{\phi}\right)$ be a smooth function in Boyer-Lindquist coordinates and assume that $\varphi$ is also $L^2$ in time. We have the Fourier-transformation of $\varphi$,
\begin{align}
\widehat{\varphi} \left(\omega,r,\theta,\tilde{\phi}\right) = \frac{1}{\sqrt{2\pi}} \int_{-\infty}^{\infty}  dt \ e^{i\omega t} \varphi \left(t,r, \theta,\tilde{\phi}\right) \, .
\end{align}
We can, for each $\omega$ and $r$, interpret $\widehat{\varphi}$ as a function on the Boyer-Lindquist spheres $\mathbb{S}^2$. There is an orthonormal basis of $L^2$-eigenfunctions on these spheres, which is obtained as the solutions of the eigenvalue problem associated with (\ref{def:pop}) and which we denote by $S_{m \ell} \left(\theta,\tilde{\phi}\right) = S_{m\ell} \left(a\omega, \cos \theta\right) e^{im\tilde{\phi}}$. Any (say smooth) function on the sphere can be expanded in terms of these eigenfunctions. Completeness and orthogonality of this basis yield:
\begin{align}
\widehat{\varphi} \left(\omega,r,\theta,\tilde{\phi}\right) = \sum_{ml} \widehat{\varphi}^{(a\omega)}_{m \ell} \left(r\right)S_{m \ell} \left(\theta,\tilde{\phi}\right) \, ,
\end{align}
\begin{align}
\widehat{\varphi}^{(a\omega)}_{m \ell} \left(r\right) = \int_{\mathbb{S}^2} \sqrt{\gamma} d\theta d\tilde{\phi} \ \widehat{\varphi} \left(\omega,r,\theta,\tilde{\phi}\right) S^\star_{m \ell} \left(\theta,\tilde{\phi}\right) \, ,
\end{align}
where $\sqrt{\gamma} = \sin \theta$ denotes the volume element with respect to which the orthogonality relation of the $S_{m \ell}$ holds, cf. again (\ref{def:pop}). We have the important Parseval identity on $\mathbb{S}^2$,
\begin{align}
\int_{\mathbb{S}^2} \sqrt{\gamma} d\theta d\tilde{\phi} \  |\widehat{\varphi} \left(\omega,r,\theta,\tilde{\phi}\right)|^2 = \sum_{m \ell} |\widehat{\varphi}^{(a\omega)}_{m \ell} \left(r\right)|^2 \, .
\end{align}
Note that both sides depend on $\omega$. 

In the paper, we restrict to low modes. More precisely, we define for given $L$, the quantity
\begin{align}
\widehat{\varphi}^{(a\omega)}_{m \ell} \Big|_{\flat}\left(r\right) = \widehat{\varphi}^{(a\omega)}_{m \ell} \left(r\right) \cdot \zeta \left(\frac{a^2\omega^2 + \lambda_{m\ell}}{L}\right)
\end{align}
and the associated quantity
\begin{align}
\widehat{\varphi}_\flat \left(\omega,r,\theta,\tilde{\phi}\right) = \sum_{ml} \widehat{\varphi}^{(a\omega)}_{m \ell} |_\flat \left(r\right)S_{m \ell} \left(\theta,\tilde{\phi}\right) \, .
\end{align}
Then, again by Parseval and the fact that $|\zeta|\leq 1$, we have
\begin{align}
\int_{\mathbb{S}^2}  \sqrt{\gamma} d\theta d\tilde{\phi} \  |\widehat{\varphi}_\flat  \left(\omega,r,\theta,\tilde{\phi}\right)|^2 = \sum_{m \ell} |\widehat{\varphi}^{(a\omega)}_{m \ell} |_\flat \left(r\right)|^2 \nonumber \\
\leq \sum_{m \ell} |\widehat{\varphi}^{(a\omega)}_{m \ell} \left(r\right)|^2 = \int_{\mathbb{S}^2}  \sqrt{\gamma} d\theta d\tilde{\phi} \  |\widehat{\varphi} \left(\omega,r,\theta,\tilde{\phi}\right)|^2 \, .\nonumber
\end{align}
Integrating this identity in $\omega$ and using that $\sqrt{\gamma}$ is independent of $\omega$, we obtain (using Parseval in $t$ and $\omega$) that
\begin{align} \label{jb}
 \int_{-\infty}^{\infty} dt \int_{\mathbb{S}^2}  \sqrt{\gamma} d\theta d\tilde{\phi} \  |\varphi_\flat  \left(t,r,\theta,\tilde{\phi}\right)|^2 \nonumber \\
\leq  \int_{-\infty}^{\infty} dt \int_{\mathbb{S}^2}  \sqrt{\gamma} d\theta d\tilde{\phi} \  |\varphi  \left(t,r,\theta,\tilde{\phi}\right)|^2 \, .
\end{align}
Of course, the same identity holds for $\varphi_\sharp$. 

\subsection*{Applications}

\begin{itemize}
\item Letting $\varphi = \psi^\tau$ we can apply the estimate (\ref{jb}) for fixed $r>r_+$ corresponding to $R^\star_{-\infty}$. In the limit $R^\star_{-\infty} \rightarrow -\infty$ one obtains an integral over the event horizon:
\begin{align}
\lim_{R^\star_{-\infty} \rightarrow -\infty} \int_{-\infty}^{\infty} dt \int_{\mathbb{S}^2}  \sqrt{\gamma} d\theta d\tilde{\phi} \  |\psi^\tau_\flat  \left(t,r\left(R^\star_{-\infty}\right),\theta,\tilde{\phi}\right)|^2 \nonumber \\
\leq \lim_{R^\star_{-\infty} \rightarrow -\infty} \int_{-\infty}^{\infty} dt \int_{\mathbb{S}^2}  \sqrt{\gamma} d\theta d\tilde{\phi} \  |\psi^\tau  \left(t,r\left(R^\star_{-\infty}\right),\theta,\tilde{\phi}\right)|^2 \nonumber \\
\leq \int_{\mathcal{H}\left(0,\tau\right)}  \sqrt{\gamma} dt^\star d\theta d\tilde{\phi} \  |\psi^\tau  \left(t,r_+,\theta,\tilde{\phi}\right)|^2 \, .
\end{align}

\item Letting $\varphi = F$, the inhomogeneity in the wave equation for $\psi^\tau_\flat$, which is compactly supported in time (in two strips, actually). Integrating (\ref{jb}) over $\int_{R^\star_{-\infty}}^{R^\star_{\infty}}  dr^\star$, we obtain that 
\begin{align}
\int_{\mathcal{R}} r^2 | F_\flat|^2 \leq \int_{\mathcal{R}} r^2 | F|^2 = \int_{\mathcal{R}\left(0,1\right) \cup \mathcal{R} \left(\tau, \tau+1\right)} r^2 | F|^2 \lesssim \| \psi \|^2_{H^1_{AdS}} \, .
\end{align}

\item The estimate (\ref{jb}) also holds for derivatives. For instance,
\begin{align} %\label{regd}
 \int_{-\infty}^{\infty} dt \int_{\mathbb{S}^2}  \sqrt{\gamma} d\theta d\tilde{\phi} \  |\left(-\partial_t + \partial_{r^\star} - \frac{a\Xi}{r_+^2 + a^2} \partial_{\tilde{\phi}}\right) \varphi_\flat |^2 \nonumber \\
 = \int_{-\infty}^\infty d\omega \sum_{m \ell} \Big| \left[\left(-\partial_t + \partial_{r^\star} - \frac{a\Xi}{r_+^2 + a^2} \partial_{\tilde{\phi}}\right) \varphi_\flat \right]^{a\omega}_{m \ell} \Big|^2
 \nonumber \\
  = \int_{-\infty}^\infty d\omega \sum_{m \ell} \Big| i \omega \widehat{\varphi}^{(a\omega)}_{m \ell} |_\flat + \partial_{r^\star} \widehat{\varphi}^{(a\omega)}_{m \ell} |_\flat  - i \frac{a\Xi m  }{r_+^2 + a^2} \widehat{\varphi}^{(a\omega)}_{m \ell} |_\flat \Big|^2 
 \nonumber \\
= \int_{-\infty}^\infty d\omega \sum_{m \ell}  \left[\zeta  \left(\frac{a^2 \omega^2 +  \lambda_{m \ell}}{L}\right) \right]^2 \Big| i \omega \widehat{\varphi}^{(a\omega)}_{m \ell}  + \partial_{r^\star} \widehat{\varphi}^{(a\omega)}_{m \ell}  - i \frac{a\Xi m }{r_+^2 + a^2} \widehat{\varphi}^{(a\omega)}_{m \ell}  \Big|^2
 \nonumber \\
\leq  \int_{-\infty}^\infty d\omega \sum_{m \ell}  \Big| i \omega \widehat{\varphi}^{(a\omega)}_{m \ell}  + \partial_{r^\star} \widehat{\varphi}^{(a\omega)}_{m \ell}  - i \frac{a\Xi m }{r_+^2 + a^2} \widehat{\varphi}^{(a\omega)}_{m \ell}  \Big|^2 \nonumber \\
 = \int_{-\infty}^\infty d\omega \sum_{m \ell} \Big| \left[\left(-\partial_t + \partial_{r^\star} - \frac{a\Xi}{r_+^2 + a^2} \partial_{\tilde{\phi}}\right) \varphi \right]^{a\omega}_{m \ell} \Big|^2 \nonumber \\
=   \int_{-\infty}^{\infty} dt \int_{\mathbb{S}^2}  \sqrt{\gamma} d\theta d\tilde{\phi} \  |\left(-\partial_t + \partial_{r^\star} - \frac{a\Xi}{r_+^2 + a^2} \partial_{\tilde{\phi}}\right) \varphi |^2  \, . \nonumber
\end{align}
Higher derivatives can be dealt with similarly. We summarize this as 
\begin{lemma}
Let $D$ denote any operator of the collection $\partial_t$, $\partial_{\tilde{\phi}}$, $\partial_\theta$, $\partial_{r^\star}$ and $Z$ (defined in (\ref{Zdef})) of Boyer-Lindquist derivative  operators. For a non-negative integer $k$ let $D^k$ denote any $k^{th}$-order combination of these derivatives. Then, for any $k$,
\begin{align} \label{alem}
 \int_{-\infty}^{\infty} dt \int_{\mathbb{S}^2}  \sqrt{\gamma} d\theta d\tilde{\phi} \  |D^k \varphi_\flat |^2 \leq  \int_{-\infty}^{\infty} dt \int_{\mathbb{S}^2}  \sqrt{\gamma} d\theta d\tilde{\phi} \  |D^k \varphi |^2 \, .
\end{align}
\end{lemma}
\begin{corollary}
Letting $\varphi = \psi^\tau$, we have for any non-negative weight function $w\left(r\right)$ the estimate
\begin{align}
\int_{\mathcal{D}}  w\left(r\right) \sqrt{\gamma} dt dr^\star d\theta d\tilde{\phi} \  \Big|D^k \psi^\tau_\flat \Big|^2 
\leq 
\int_{\mathcal{D}} w\left(r\right) \sqrt{\gamma} dt dr^\star d\theta d\tilde{\phi} \ \Big|D^k \psi^\tau \Big|^2  \, .
\end{align}
\end{corollary} 
We may now choose $w\left(r\right)$ with the correct weights near the horizon and near infinity to mimic the volume element $\sqrt{g}$. Since the derivatives $D$ are all regular, and can readily be converted into the regular $\left(t^\star, r, \theta, \phi\right)$- coordinates, we finally obtain (using that the right hand side is only supported for $0\leq t^\star \leq \tau$):
\begin{align} \label{e2est}
\int_{\mathcal{D}}  \sqrt{g} dt^\star dr d\theta d\phi \  e_2\left[\psi^\tau_\flat\right] \leq \int_{\mathcal{R}\left(0,\tau\right)}  \sqrt{g} dt^\star dr d\theta d\phi \  e_2\left[\psi^\tau\right] \, 
\leq E_2\left[\psi\right] \tau \, .
\end{align}
The same estimate holds for $\psi^\tau_\sharp$.
\end{itemize}

%%      Type body of appendix/-ices here.

%%      ---------------------------------------------------------------------
%%      ---------------------------ACKNOWLEDGMENTS (OPTIONAL) ---------------
%%      ---------------------------------------------------------------------

%% ***** UNCOMMENT THE FOLLOWING LINE TO ADD ACKNOWLEDGMENTS.

%%      Type acknowledgments here.

%%      ---------------------------------------------------------------------
%%      --------------------------- BIBLIOGRAPHY ----------------------------
%%      ---------------------------------------------------------------------

\frenchspacing
\bibliographystyle{hacm}
\bibliography{thesisrefs}
\end{document}